\documentclass[11pt]{article}

\usepackage{bm, commath}
\usepackage{natbib}
\usepackage{caption}
\usepackage{graphicx}
\usepackage{subfig}
\usepackage{amsmath, amsfonts, amsthm}
\usepackage{float}
\usepackage{booktabs,siunitx}
\usepackage{url}
\usepackage{multirow}
\usepackage{amssymb}
\usepackage{blkarray}
\usepackage{bbm}
\usepackage{multicol}
\usepackage{authblk}
\usepackage[shortlabels]{enumitem}

\usepackage{listings}
\usepackage{bbm}
\usepackage{textcomp}
\usepackage[margin=1in]{geometry}
\usepackage{authblk}
\usepackage[ruled]{algorithm2e}
\SetKwInput{KwParam}{Parameter}
\SetAlgoCaptionLayout{centerline}

\usepackage{sectsty}
\usepackage[doublespacing]{setspace}
\usepackage[utf8]{inputenc}
\usepackage{float}
\usepackage{tikz}
\usetikzlibrary{shapes,decorations,arrows,calc,arrows.meta,fit,positioning}

\newtheorem{definition}{Definition}
\newtheorem{corollary}{Corollary}

\newtheorem{assumption}{Assumption}
\newtheorem*{assumption*}{Assumption}
\newtheorem{side_info}{Side-Information}
\newtheorem{proposition}{Proposition}

\newtheorem{lemma}{Lemma}
\newtheorem{constraint}{Constraint}

\usepackage{mathtools}

\usepackage{xcolor}

\makeatletter
\renewcommand{\algocf@captiontext}[2]{#1\algocf@typo. \AlCapFnt{}#2} 
\def\@algocf@capt@plain{top}
\renewcommand{\algocf@makecaption}[2]{%
  \addtolength{\hsize}{\algomargin}%
  \sbox\@tempboxa{\algocf@captiontext{#1}{#2}}%
  \ifdim\wd\@tempboxa >\hsize
  \hskip .5\algomargin%
  \parbox[t]{\hsize}{\algocf@captiontext{#1}{#2}}
  \else%
  \global\@minipagefalse%
  \hbox to\hsize{\box\@tempboxa}
  \fi%
  \addtolength{\hsize}{-\algomargin}%
}
\makeatother

\begin{document}

\sectionfont{\bfseries\large\sffamily}%

\subsectionfont{\bfseries\sffamily\normalsize}%




\title{Testing Biased Randomization Assumptions and Quantifying Imperfect Matching and Residual Confounding in Matched Observational Studies}

\author[1]{Kan Chen }
\author[2]{Siyu Heng }
\author[3]{Qi Long }
\author[4]{Bo Zhang \thanks{Email: {\tt bzhang3@fredhutch.org}} }

\affil[1]{Graduate Group of Applied Mathematics and Computational Science, School of Arts and Sciences, University of Pennsylvania, Philadelphia, Pennsylvania, U.S.A.}
\affil[2]{Department of Biostatistics, School of Global Public Health, New York University, New York City, New York, U.S.A.}
\affil[3]{Department of Biostatistics, Epidemiology and Informatics, Perelman School of Medicine, University of Pennsylvania, Philadelphia, Pennsylvania, U.S.A.}
\affil[4]{Vaccine and Infectious Disease Division, Fred Hutchinson Cancer Center, Seattle, Washington, U.S.A.}

\date{}

\maketitle

\noindent
\textsf{{\bf Abstract}: One central goal of design of observational studies is to embed non-experimental data into an approximate randomized controlled trial using statistical matching. Despite empirical researchers' best intention and effort to create high-quality matched samples, residual imbalance due to observed covariates not being well matched often persists. Although statistical tests have been developed to test the randomization assumption and its implications, few provide a means to quantify the level of residual confounding due to observed covariates not being well matched in matched samples. In this article, we develop two generic classes of exact statistical tests for a biased randomization assumption. One important by-product of our testing framework is a quantity called \emph{residual sensitivity value} (RSV), which provides a means to quantify the level of residual confounding due to imperfect matching of observed covariates in a matched sample. We advocate taking into account RSV in the downstream primary analysis. The proposed methodology is illustrated by re-examining a famous observational study concerning the effect of right heart catheterization (RHC) in the initial care of critically ill patients. Code implementing the method can be found in the supplementary materials.}%

\vspace{0.3 cm}

\section{Introduction}
\label{sec: intro}

\subsection{Statistical matching and randomization-based outcome analysis}
In an observational study of a treatment's effect on the outcome, units involved in the study may differ systematically in their observed pretreatment covariates, thus invalidating a naive comparison between the treated and control groups. Statistical matching and subclassification are commonly used nonparametric tools to adjust for covariates. For a binary treatment, the ultimate goal of statistical matching is to embed observational data into an approximate randomized experiment by designing a treated group and a matched control (or comparison) group that are comparable in their observed covariates (\citealp{rosenbaum2002observational,rosenbaum2010design, stuart2010matching, bind2019bridging}).

One widely-used downstream outcome analysis method for matched data is randomization-based inference as if the matched data comes from a randomized controlled trial (RCT). In our view, this practice is appealing for two reasons. First, randomization inference (possibly with regression adjustment) is commonly used in analyzing RCT data; hence, using a randomization-based procedure to analyze matched data fits into the big picture of statistical matching, i.e., to embed observational data into an experiment. Second, the sensitivity analysis framework complementing the primary randomization-based outcome analysis is well understood and developed. The Rosenbaum-bounds-type sensitivity analysis (\citealp{rosenbaum2002observational,rosenbaum2010design, diprete2004}) allows the treatment assignment probability in each matched pair (or set) to deviate from the randomization probability up to a degree controlled by a parameter $\Gamma$ and outputs a bounding $p$-value under such deviation. Furthermore, randomization-based inferential methods can be applied to testing both Fisher's sharp null hypothesis (\citealp{rosenbaum2002observational, rosenbaum2010design}) and Neyman's weak null hypothesis (\citealp{wu2020randomization}). Sensitivity analysis methods have also been developed for both sharp and weak null hypotheses (\citealp{rosenbaum2002observational, rosenbaum2010design, fogarty2020studentized}).

\subsection{The randomization assumption}
In a typical randomization-based downstream outcome analysis of matched-pair data, researchers make the following \emph{Randomization Assumption} (\citealp{rosenbaum2002observational, rosenbaum2010design}):
\begin{assumption*}[Randomization Assumption, Stated Informally]
The treatment assignments across all matched pairs are assumed to be independent of each other and that the treatment is randomly assigned in each matched pair, i.e., probability of the first unit receiving treatment and the other control is the same as the second unit receiving treatment and the first control. 
\end{assumption*}

The randomization assumption (RA) entails two aspects: (i) independence among matched pairs and (ii) random treatment assignment in each matched pair. It is the single most important assumption as it enables researchers to treat observational data after statistical matching as if the data were from a randomized controlled experiment (\citealp{dehejia2002propensity, rosenbaum2002observational,rosenbaum2010design}). Researchers make analogous assumptions when analyzing observational data after 1-to-$k$ matching, full matching (\citealp{heng2019increasing,zhang2022statistical}), and matched-pair clustered designs (\citealp{hansen2014clustered, zhang2020bridging}). 

The RA clearly holds by design in an RCT. It also holds if matched pairs are independent of each other and two units in the same pair have the same propensity score. In practice, the true propensity score for each unit is at best known up to estimation uncertainty if one happens to correctly specify the propensity score model and at worst unfathomable. Moreover, the matching algorithm may create weak dependence among matched pairs or sets so that the randomization assumption is at best an approximation of the reality.

Researchers perform randomization-based outcome analysis for assorted matched data, not restricted to propensity-score-matched (PSM) data. In fact, as shown by \citet{rosenbaum1985constructing}, propensity score matching often underperforms more sophisticated matching algorithms that combine metric-based and propensity-score-based matching in balancing observed covariates. Indeed, many modern statistical matching algorithms have moved beyond PSM; some notable examples include network-flow-based optimal matching and its many variants (\citealp{rosenbaum1989optimal, rosenbaum2002observational, rosenbaum2010design,zhang2021matching}, among others), coarsened exact matching (\citealp{iacus2011multivariate}), mixed-integer-programming-based algorithms (\citealp{zubizarreta2012using}), and genetic matching algorithm (\citealp{diamond2013genetic}), among others.

\subsection{Justifications for randomization inference: Informal and formal diagnostics}
\label{subsec: intro diagnostics of RA}
To justify using a randomization-based inferential procedure, researchers often perform informal diagnostics based on metrics like the standardized mean differences (SMDs) to examine the covariate balance after matching (\citealp{silber2001multivariate}, \citealp{franklin2014metrics}, \citealp{austin2015moving}) or formal statistical procedures to test the equality of covariate distributions in the treated and matched control groups (e.g., \citet{rosenbaum2005exact}'s crossmatch test). However, there is a gap between equality of covariate distributions and the RA: the downstream randomization inference relies solely on the RA and need not assume sampling covariates or potential outcomes from some superpopulation. A more detailed review and related discussion can be found in the Supplementary Material A.

More recently, \citet{gagnon2019classification} developed the Classification Permutation Test (CPT) that can be adapted to testing the RA as follows. First, train a classifier $\widehat{f}$ using any classification tools (e.g., logistic regression, support vector machine, random forests, or an ensemble of them) with treatment status $\mathbf{Z}$ as the label and observed covariates $\mathbf{X}$ as predictors, and predict treatment indicators for all matched-pair data using $\widehat{f}$. Denote the predicted treatment indicators by $(\widehat{Z}_{11}, \dots, \widehat{Z}_{I2})$ and define the test statistic $T=\sum_{i=1}^{I}\sum_{j=1}^{2} Z_{ij}\widehat{Z}_{ij}$. Next, for $m = 1, \cdots, \text{MC}$, randomly permute two treatment indicators within each pair, denoted as $(Z_{11}^{(m)}, \dots, Z_{I2}^{(m)})$, re-train the classifier $\widehat{f}^{(m)}$ using the same covariates data $\mathbf{X}$ but permuted treatment indicators $(Z_{11}^{(m)}, \dots, Z_{I2}^{(m)})$, and then re-predict treatment assignments using $\widehat{f}^{(m)}$. Denote by $(\widehat{Z}_{11}^{(m)}, \dots, \widehat{Z}_{I2}^{(m)})$ the prediction at iteration $m$. To test the randomization assumption, it then suffices to compare the test statistic $T$ to the null distribution generated by $\{T^{(m)}=\sum_{i=1}^{I}\sum_{j=1}^{2} Z_{ij}^{(m)}\widehat{Z}_{ij}^{(m)},~m = 1, \cdots, \text{MC}\}$. In this way, the CPT procedure yields an exact $p$-value for testing the randomization assumption. Similar Fisherian-style permutation strategies are also used in \citet{branson2018randomization} and \citet{branson2020evaluating} to deliver an exact test for the RA.

\subsection{Moving beyond the randomization assumption: Testing a biased randomization assumption and quantifying residual confounding}
A powerful test of the RA is desirable as it closely examines a most important premise for downstream statistical inference; however, a test of the RA by itself may not be comprehensive enough to capture the full picture of the study design. Take as an example \citet{gagnon2019classification}'s re-analysis of \citet{heller2010using}'s matched-pair data. \citet[Section 4.4]{gagnon2019classification} rejected the RA at the $0.05$ level using the CPT procedure and concluded that ``the covariates can predict the treatment assignment better than under random assignment." As one carefully examines the covariate balance of \citet{heller2010using}'s matched-pair data, however, no two-sample t-test, Wilcoxon signed rank test, or Kolmogorow-Smirnov test is statistically significant at the $0.1$ level for any of the observed covariates in the treated and matched control groups, suggesting that there is likely to be little residual confounding due to imperfect matching on observed covariates. It is also unclear to what extent a minor deviation from the RA would affect the downstream outcome analysis. From a practical perspective, with a dataset as well-matched as that in \citet{heller2010using}, few empirical researchers would re-do the sometimes computationally-intensive statistical matching. A powerful test of the RA may thus have an unintended consequence of discouraging empirical researchers from adopting and reporting such formal diagnostics. In our opinion, what currently are lacking in the literature are three-fold: (i) a measure to quantify the level of residual imbalance in observed covariates, an arguably more important practical matter than testing and rejecting the RA by itself, (ii) a means to systematically incorporate ``recalcitrant" residual covariate imbalance, if there is any, into the outcome analysis, and (iii) simulation results that relate the quality of statistical matching to the performance of outcome analysis.

\subsection{Our contribution}
\label{subsec: intro contribution}
This article aims to start filling in these gaps. We propose two generic classes of exact statistical tests for a relaxed version of the RA, termed the biased randomization assumption (biased RA). Our first proposal is based on a sample-splitting version of the CPT (termed SS-CPT). Our second proposal, clustering-based test (termed CBT), takes a different perspective and recasts the testing problem as a clustering problem with side information. Next, we propose to use a one-number summary statistic, termed the \emph{residual sensitivity value (RSV)}, to \emph{quantify} the deviation from the RA due to residual imbalance in observed covariates by inverting a nested sequence of statistical tests of the biased RA. We then describe how to incorporate the RSV into the downstream, randomization-based outcome analysis. Via extensive simulation studies, we systematically compare the proposed tests, and examine how the quality of statistical matching meaningfully affects the statistical performance of the randomization-based outcome analysis; in particular, our simulation results suggest that when the RA cannot be rejected by proposed tests, the downstream randomization-based inference has desirable statistical performance in the sense that the Hodges-Lehmann point estimate (\citealp[Chapter 2]{rosenbaum2002observational}) has small mean squared error. We also investigate the performance of statistical inference after incorporating the RSVs via simulation studies.

The article is organized as follows. Section \ref{sec: RA and biased RA} formally defines the biased RA. Section \ref{sec: SS-CPT} and \ref{sec: CBT} introduce two generic classes of statistical tests for the biased RA. Section \ref{sec: simulation} presents simulation results. Section \ref{sec: application RHC} considers a case study. We conclude with a discussion in Section \ref{sec: discussion}. 

\section{Randomization and biased randomization assumption}
\label{sec: RA and biased RA}
We consider the setting of a typical matched cohort study. Suppose that the study has access to $N_t$ treated units and a large reservoir of $N_c$ control units so that there are $N_t + N_c = N$ units in total. Without loss of generality, we assume $N_c > N_t$ (if not, switch the role of treated and control units) and $N_c$ is often much larger than $N_t$. Each of the $N$ units is associated with a treatment indicator $Z_n$ and a $q$-dimensional vector of observed covariates $\mathbf{x}_n$, $n = 1, \dots, N$. Researchers match each treated unit with a control unit using any statistical matching procedure $\textsf{M}$ and produce $I \leq N_t$ matched pairs of two units, one treated and the other control. We use $ij$, $i = 1, \dots, I,~j = 1, 2$, to index the $j$-th unit in the $i$-th matched pair. Let $\mathbf{x}_{ij}$ denote unit $ij$'s observed covariates and $Z_{ij}$ its treatment status so that the unit with $Z_{ij} = 1$ is the treated unit and with $Z_{ij} = 0$ the control unit. Finally, we collect the observed covariates information of $2I$ matched units in the matrix $\mathbf{X}$. We assume no unmeasured confounders.

The randomization assumption in a downstream, finite-sample, randomization-based outcome analysis can be formally stated as follows:
\begin{assumption}[Randomization Assumption in Matched-Pair Studies, Stated Formally]\label{assum: randomization inference}
Treatment assignments across matched pairs are assumed to be independent of each other, with
\begin{equation}
\label{eqn: assumption 1}
    P(Z_{i1}=1, Z_{i2}=0\mid \mathbf{X}, \textsf{M})=P(Z_{i1}=0, Z_{i2}=1 \mid \mathbf{X}, \textsf{M})=1/2,\quad \text{for}~~i = 1, \cdots, I.
\end{equation}
\end{assumption}

Statistical matching can largely remove overt bias from observed covariates (\citealp{rosenbaum2002observational,rosenbaum2010design}); however, some residual confounding from observed covariates may persist in the matched sample due to imperfect matching, and this motivates a biased randomization assumption:
\begin{assumption}[Biased Randomization Assumption]\label{assum: biased randomization inference}
Treatment assignments across matched pairs are assumed to be independent of each other, with
\begin{equation}
\label{eqn: assumption biased RA}
\Gamma^{-1} \leq \frac{ P(Z_{i1}=1, Z_{i2}=0\mid \mathbf{X}, \textsf{M})}{P(Z_{i1}=0, Z_{i2}=1 \mid \mathbf{X}, \textsf{M})} \leq \Gamma,\quad \text{for}~~i= 1, \dots, I, ~\text{and some}~\Gamma \in [1, \infty).
\end{equation}
\end{assumption}

Assumption \ref{assum: biased randomization inference} is the basis for the Rosenbaum-bounds sensitivity analysis framework (\citealp{rosenbaum2002observational, rosenbaum2010design}). It uses one parameter $\Gamma$ to control the maximum degree to which the residual confounding biases the treatment assignment probability in each matched pair. We stress again that we assume no unmeasured confounding in this article and are only interested in assessing residual confounding from observed covariates. In the presence of unmeasured confounding, $\Gamma$ can be \emph{arbitrarily} large, and Assumption \ref{assum: biased randomization inference} becomes untestable.

\section{Sample-splitting classification permutation test (SS-CPT)}
\label{sec: SS-CPT}
Our first proposal to test Assumption \ref{assum: biased randomization inference} is a simple modification of the CPT. First, randomly split $I$ matched pairs into two non-overlapping index sets $\mathcal{I}^{(1)}$ and $\mathcal{I}^{(2)}$, and train a classifier $\hat{f}_{1}$ using observed covariates $\mathbf{X}^{(1)}=\{ X_{ij}, i \in \mathcal{I}^{(1)}, j=1,2\}$ as predictors and treatment status $\mathbf{Z}^{(1)}=\{ Z_{ij}, i \in \mathcal{I}^{(1)}, j=1,2\}$ as labels. Next, apply the classifier $\hat{f}_{1}$ to the observed covariates $\mathbf{X}^{(2)}=\{ X_{ij}, i \in \mathcal{I}^{(2)}, j=1,2\}$ and let $\hat{\mathbf{f}}_{1\rightarrow2}=\{\hat{f}(\mathbf{x}_{ij}), i\in \mathcal{I}^{(2)}, j=1,2\} \in [0,1]^{|\mathcal{I}^{(2)}|}$ denote the predicted scores. Lastly, let $g_{ij}: [0,1]^{|\mathcal{I}^{(2)}|} \rightarrow \mathbb{R}$ denote a generic function and form the test statistic $T_{1\rightarrow 2}=\sum_{i\in \mathcal{I}^{(2)}} \sum_{j=1}^{2}Z_{ij}g_{ij}(\hat{\mathbf{f}}_{1\rightarrow2})$, which can then be used to test the biased RA restricted to matched pairs $i \in \mathcal{I}^{(2)}$ for a fixed $\Gamma$ value according to Proposition \ref{prop: null distribution biased for CPT} below.

\begin{proposition}[Bounding $p$-value]\label{prop: null distribution biased for CPT}
Let $T$ be a test statistic of the form $T=\sum_{i=1}^{I}\sum_{j=1}^{2}Z_{ij}q_{ij}$, where $q_{ij}$ is some fixed score based on the observed covariates $\mathbf{X}=\{\mathbf{x}_{ij}: i=1,\dots, I, j = 1, 2\}$. For $i=1,\dots,I$, define $\overline{T}_{\Gamma i}$ to be independent random variables taking the value $\overline{\overline{q}}_{i}= \max\{q_{i1}, q_{i2}\}$ with probability $\Gamma/(1+\Gamma)$ and the value $\overline{q}_{i}=\min\{q_{i1}, q_{i2}\}$ with probability $1/(1+\Gamma)$. Under Assumption~\ref{assum: biased randomization inference} with $\Gamma \geq 1$, we have for any $t$,
\begin{align*}
P\left(T \geq t \mid \mathbf{X}, \textsf{M}\right)\leq P\Big(\sum_{i=1}^{I}\overline{T}_{\Gamma i} \geq t \Big)\overset{\Delta}{=}\overline{p}_{\Gamma,\textsf{exact}}\simeq 1-\Phi\left(\frac{t-\sum_{i=1}^{I}(\frac{\Gamma}{1+\Gamma}\cdot\overline{\overline{q}}_{i}+\frac{1}{1 + \Gamma }\overline{q}_{i})}{\sqrt{\sum_{i=1}^{I}\frac{\Gamma}{(1+\Gamma)^{2}}(\overline{\overline{q}}_{i}-\overline{q}_{i})^{2}} }\right)\overset{\Delta}{=}\overline{p}_{\Gamma,\textsf{approx}},
\end{align*}
where $\Phi(\cdot )$ is the distribution function of standard normal distribution, and ``$\simeq$" denotes that two sequences are asymptotically equal as $I \rightarrow \infty$.
\end{proposition}

\begin{proof}
All proofs in the article can be found in the Supplementary Material B.
\end{proof}

In Proposition \ref{prop: null distribution biased for CPT}, let $I = |\mathcal{I}^{(2)}|$ and $q_{ij} = g_{ij}(\hat{\mathbf{f}}_{1\rightarrow2})$, and we can then calculate the bounding $p$-value $p_{1\rightarrow 2, \Gamma}$ under the biased RA restricted to $i \in \mathcal{I}^{(2)}$ for a fixed $\Gamma$ value. This bounding $p$-value can be obtained by calculating the tail probability of the random variable $\sum_{i=1}^{I}\overline{T}_{\Gamma i}$ via Monte Carlo or via the Normal approximation. Next, flip the role of $\mathcal{I}^{(1)}$ and $\mathcal{I}^{(2)}$, form a second test statistic $T_{2\rightarrow 1}=\sum_{i\in \mathcal{I}^{(1)}} \sum_{j=1}^{2}Z_{ij}h_{ij}(\hat{\mathbf{f}}_{2\rightarrow 1})$ where $h_{ij}: [0,1]^{|\mathcal{I}^{(1)}|} \rightarrow \mathbb{R}$ is another generic function, and obtain $p_{2\rightarrow 1, \Gamma}$ in a similar way. Finally, we reject Assumption \ref{assum: biased randomization inference} with a prespecified $\Gamma$ at the level $\alpha$ when $\overline{p}_{\Gamma} = \min\{ p_{1\rightarrow 2, \Gamma}, p_{2\rightarrow 1, \Gamma}\}<\alpha/2$.  Algorithm \ref{alg: SS-CPT} summarizes this sample-splitting variant of the CPT, which we refer to as SS-CPT. Sample-splitting is an essential element of Algorithm \ref{alg: SS-CPT} because scores $q_{ij}$ in Proposition \ref{prop: null distribution biased for CPT} are held fixed by sample-splitting; in contrary, scores in the vanilla version of the CPT reviewed in Section \ref{subsec: intro diagnostics of RA} depend on the permuted treatment labels and change at each and every permutation. It is unclear how to derive a bounding $p$-value and test the biased RA using the vanilla CPT.


What are some sensible choices of $g_{ij}$ (and similarly $h_{ij}$) in Algorithm \ref{alg: SS-CPT}? There are at least three choices. First, we may let $g_{ij}(\hat{\mathbf{f}}_{1\rightarrow2})=\mathbf{1}\{\hat{f}(\mathbf{x}_{ij})>\hat{f}(\mathbf{x}_{ij^{\prime}})\}$, and assign the predicted ``treated" label to the unit with the higher predicted score and ``control" label to the other within each matched pair. In this way, the test statistic $T$ is essentially a prediction accuracy measure described in \citet{gagnon2019classification}. Second, we may directly take the predicted score as $q_{ij}$ by setting $g_{ij}(\hat{\mathbf{f}}_{1\rightarrow2})=\hat{f}(\mathbf{x}_{ij})$. Third, we may take the rank of the predicted score for the unit $ij$ among the predicted scores of all study units as $q_{ij}$ by setting $g_{ij}(\hat{\mathbf{f}}_{1\rightarrow2})=\sum_{i^{\prime}=1}^{I}\sum_{j^{\prime}=1}^{2} \mathbf{1}(\hat{f}(\mathbf{x}_{ij})\geq \hat{f}(\mathbf{x}_{i^{\prime}j^{\prime}}))$.



\begin{algorithm}[ht]

\caption{Pseudo Algorithm for Testing Assumption~\ref{assum: biased randomization inference} using SS-CPT }
\label{alg: SS-CPT}

\textbf{Input:} $\mathbf{X}=\{\mathbf{x}_{ij}: i=1,\dots, I, j = 1, 2\}$ and $\mathbf{Z}=\{Z_{ij}: i=1,\dots, I, j = 1, 2\}$. 

\textbf{1.} Randomly split the matched sample into two parts: $\mathcal{I}^{(1)}$ and $\mathcal{I}^{(2)}$ such that $\mathcal{I}^{(1)} \cup \mathcal{I}^{(2)}=\{1,\dots, I\}$ and $\mathcal{I}^{(1)} \cap \mathcal{I}^{(2)}=\emptyset$. 

\textsf{2.} Train a classifier $\hat{f}_{1}$ with covariates $\mathbf{X}^{(1)}=\{ X_{ij}, i \in \mathcal{I}^{(1)}, j=1,2\}$ as predictors and treatment status $\mathbf{Z}^{(1)}=\{ Z_{ij}, i \in \mathcal{I}^{(1)}, j=1,2\}$ as labels. 

\textsf{3.} Train a classifier $\hat{f}_{2}$ with covariates $\mathbf{X}^{(2)}=\{ X_{ij}, i \in \mathcal{I}^{(2)}, j=1,2\}$ as predictors and treatment status $\mathbf{Z}^{(2)}=\{ Z_{ij}, i \in \mathcal{I}^{(2)}, j=1,2\}$ as labels.


\textsf{4.} Let $\hat{\mathbf{f}}_{1\rightarrow2}=\{\hat{f}_1(\mathbf{x}_{ij}), i\in \mathcal{I}^{(2)}, j=1,2\}$ denote the predicted scores for $\mathcal{I}^{(2)}$ based on the classifier $\hat{f}_{1}$. Define the test statistic $T_{1\rightarrow 2}=\sum_{i\in \mathcal{I}^{(2)} } \sum_{j=1}^{2}Z_{ij}g_{ij}(\hat{\mathbf{f}}_{1\rightarrow2})$, and calculate the worst-case $p$-value under the biased RA with $\Gamma$ according to Proposition~\ref{prop: null distribution biased for CPT}. Denote this bounding $p$-value by $p_{1\rightarrow 2, \Gamma}$.

\textsf{5.} Let $\hat{\mathbf{f}}_{2\rightarrow1}=\{\hat{f}_2(\mathbf{x}_{ij}), i\in \mathcal{I}^{(1)}, j=1,2\}$ denote the predicted scores for $\mathcal{I}^{(1)}$ based on the classifier $\hat{f}_{2}$. Define the test statistic $T_{2\rightarrow 1}=\sum_{i\in \mathcal{I}^{(1)} } \sum_{j=1}^{2}Z_{ij}h_{ij}(\hat{\mathbf{f}}_{2\rightarrow 1})$, and calculate the worst-case $p$-value under the biased RA with $\Gamma$ according to Proposition~\ref{prop: null distribution biased for CPT}. Denote this bounding $p$-value by $p_{2\rightarrow 1, \Gamma}$.

\textbf{Output:} Reject Assumption~\ref{assum: biased randomization inference} with the prespecified $\Gamma$ if and only if $\min\{ p_{1\rightarrow 2, \Gamma}, p_{2\rightarrow 1, \Gamma}\}<\alpha/2$.

\end{algorithm}

\section{Clustering-based test (CBT)}
\label{sec: CBT}
\subsection{A clustering-based framework}
If a clustering algorithm $\textsf{ALG}$ can correctly group matched-pair data into a treated cluster and a control cluster solely based on the observed covariates, then this provides evidence that data is not well-matched or well-overlapped, and some degree of deviation from the RA persists. This intuition is illustrated in Figure \ref{fig: a schematic plot}. Definition \ref{def: appropriate algorithm} states that an algorithm $\textsf{ALG}$ is appropriate for testing Assumption \ref{assum: biased randomization inference} if it leverages no more information than what equation \eqref{eqn: assumption biased RA} conditions upon, and outputs ``educated guesses" of cluster membership that respect the matched-pair structure. In particular, Definition \ref{def: appropriate algorithm} rules out using the treatment status (up to $Z_{i1} + Z_{i2} = 1$) or outcome data. 

\begin{figure}[ht]
   \centering
     \subfloat[Matched data]{\includegraphics[width = 0.49\columnwidth]{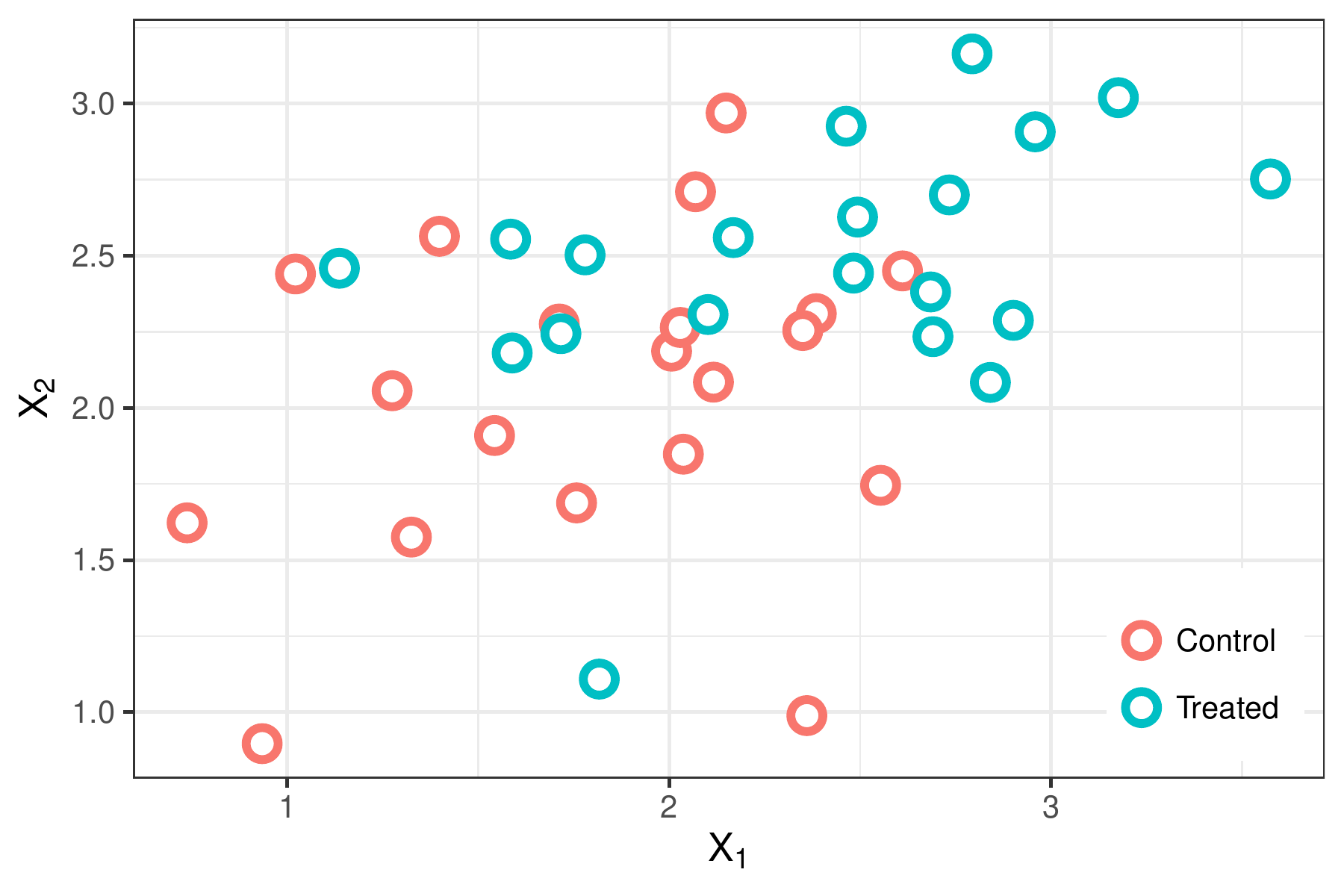}}
     \subfloat[Anonymized matched data]{\includegraphics[width = 0.49\columnwidth]{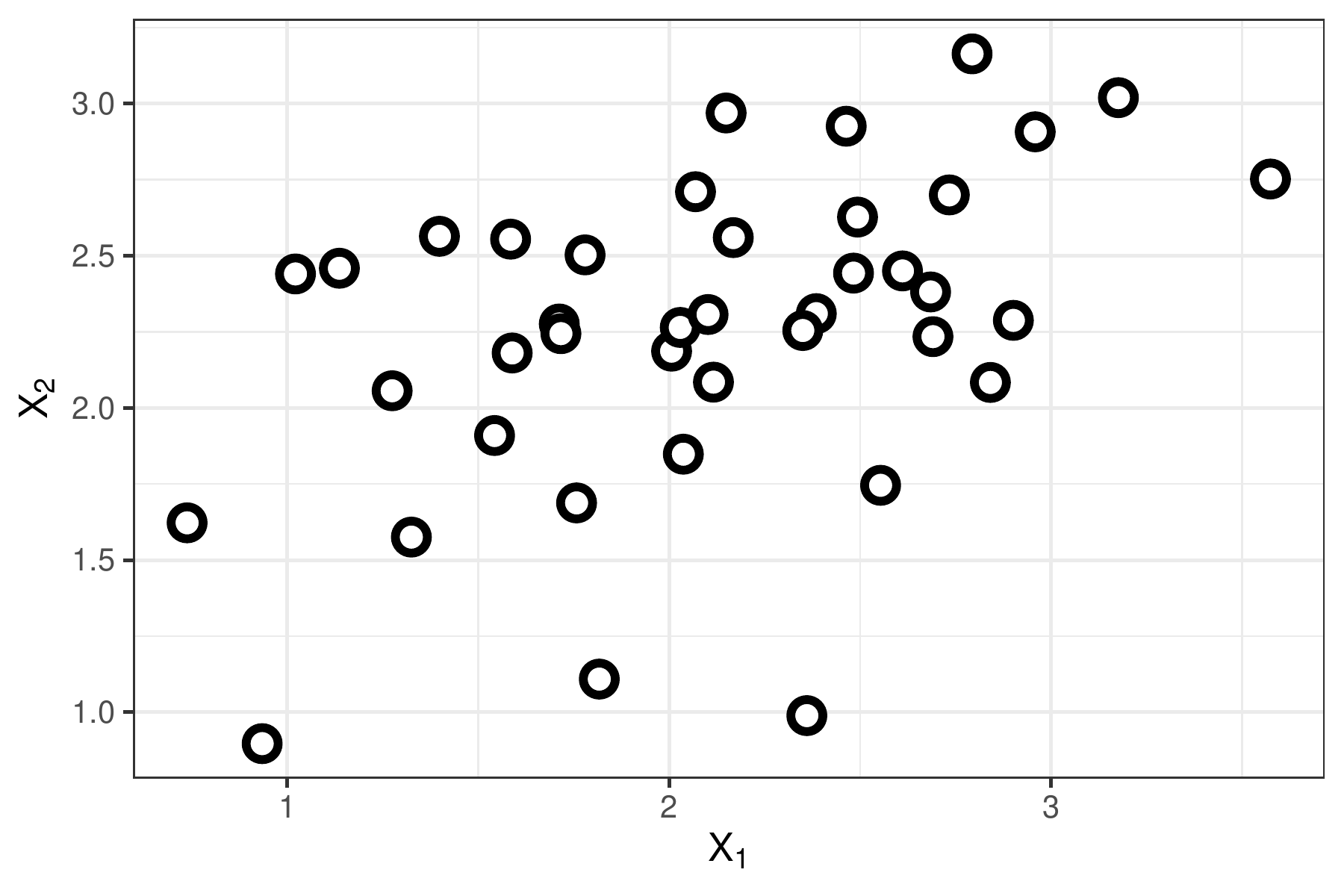}}\\
     \subfloat[Anonymized matched data after clustering]{\includegraphics[width = 0.49\columnwidth]{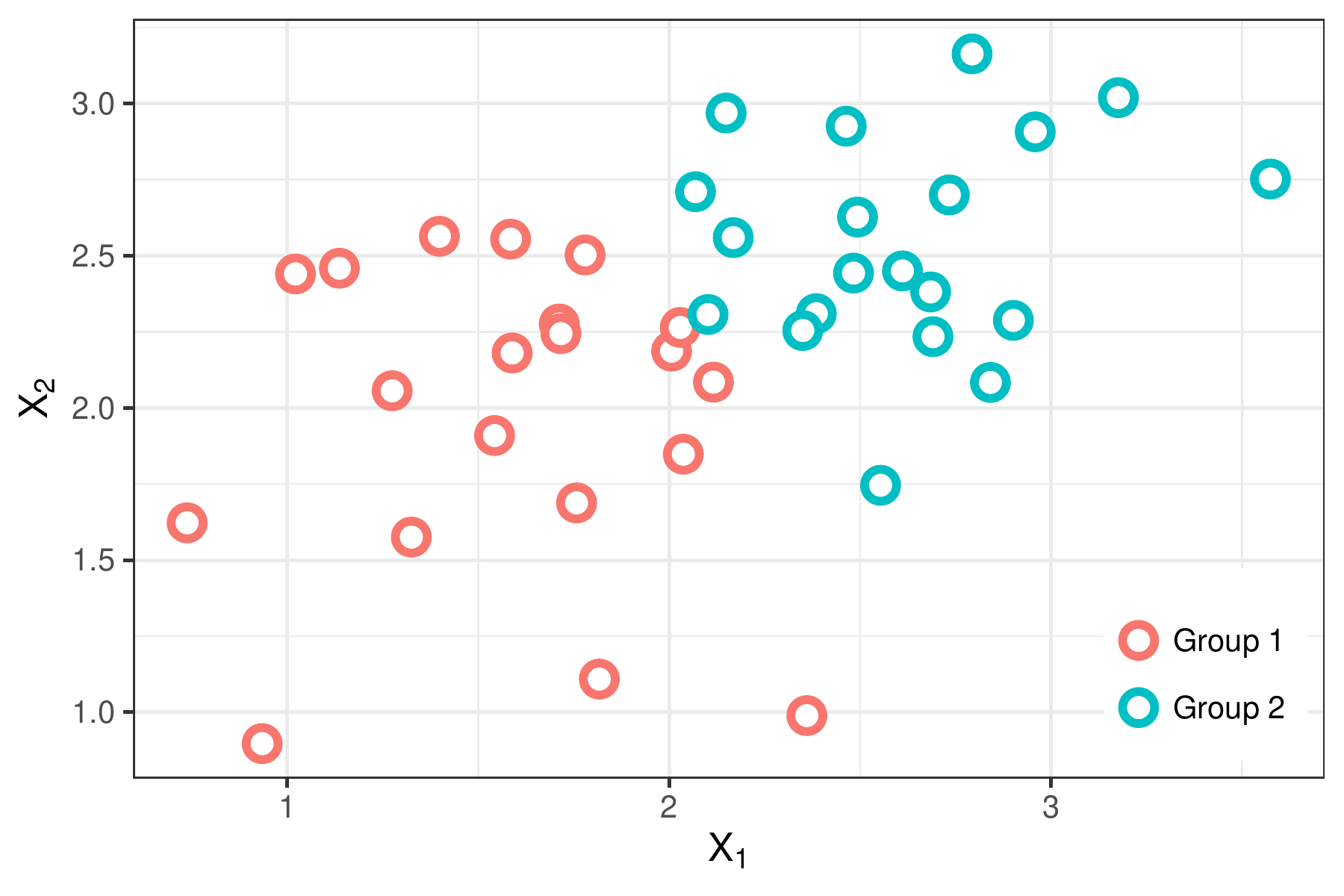}}
     \caption{\small A schematic plot with two-dimensional observed covariates $(X_1, X_2)$. Top left panel: 20 treated and 20 control subjects after matching. Top right panel: the same data with treatment labels anonymized. This is the input into the clustering algorithm. Bottom panel: anonymized data partitioned into two groups (Group 1 versus Group 2) based on a clustering algorithm. This partition is then compared to the true labels (top left panel) and high clustering accuracy is used as evidence testing the RA and biased RA.}
\label{fig: a schematic plot}
\end{figure}

\begin{definition}[Appropriateness]\label{def: appropriate algorithm}
Let $\mathcal{F} = \{\mathbf{X}, Z_{i1} + Z_{i2} = 1,\forall i = 1, \cdots, I\}$ be the collection of information. An algorithm $\textsf{ALG}$ is appropriate if it takes as input $\mathcal{F}$ and outputs a partition of $2I$ matched units into two equal-sized groups (each with size $I$): $\prod_{1}=\{ij_{1}: i=1,\dots,I\}$ and $\prod_{2}=\{ij_{2}: i=1,\dots,I\}$, subject to the constraint that there is one and only one treated unit in each matched pair $i$, i.e., we have either $i1 \in \prod_{1}$ and $i2 \in \prod_{2}$ or $i2 \in \prod_{1}$ and $i1 \in \prod_{2}$.
\end{definition}

Let $\textsf{ALG}$ be an appropriate algorithm satisfying Definition \ref{def: appropriate algorithm}. Proposition \ref{prop: null distribution biased} specifies the null distribution of the number of ``correctly guessed" cluster membership returned by $\textsf{ALG}$ under the biased RA for a fixed $\Gamma$. The null distribution under the RA is a special case of Proposition \ref{prop: null distribution biased} by setting $\Gamma = 1$ and presented as Corollary \ref{cor: null distribution} for easy reference.

\begin{proposition}[Bounding Two-Sided $p$-value]\label{prop: null distribution biased}
Let $\mathcal{T} = \{ij: Z_{ij}=1, i=1,\dots, I, j=1, 2\}$ be the set of indices corresponding to the treated units in each matched pair. Let $\Pi_1$ and $\Pi_2$ be the output from an appropriate algorithm $\textsf{ALG}$. For $i=1,\dots,I$, define $\widetilde{T}_{\Gamma i}$ to be independent random variables taking the value $1$ with probability $\Gamma/(1+\Gamma)$ and the value $0$ with probability $1/(1+\Gamma)$. Under Assumption~\ref{assum: biased randomization inference} with $\Gamma \geq 1$, we have for any $t$,
\begin{align*}
 &\quad P\left((|\Pi_{1} \cap  \mathcal{T}|-I/2)^{2} \geq (t-I/2)^{2} \mid \mathbf{X}, \textsf{M}\right)\\
 &\leq P\Big(\sum_{i=1}^{I}\widetilde{T}_{\Gamma i} \geq |t-I/2|+I/2 \Big)+P\Big(\sum_{i=1}^{I}(1-\widetilde{T}_{\Gamma i}) \leq -|t-I/2|+I/2 \Big) ~\overset{\Delta}{=}~\overline{p}_{\Gamma,\textsf{exact}} \\
 &\simeq 1-\Phi\left(\frac{|t-I/2|+I/2-\left\{\Gamma/(1 + \Gamma)\right\}\cdot I }{ \left\{\Gamma/(1+\Gamma)^2 \cdot I\right\}^{1/2}}\right) + \Phi\left(\frac{-|t-I/2|+I/2-\left\{1/(1 + \Gamma)\right\}\cdot I }{ \left\{\Gamma/(1+\Gamma)^2 \cdot I\right\}^{1/2}}\right)\\
 &\overset{\Delta}{=}\overline{p}_{\Gamma,\textsf{approx}}.
\end{align*}
\end{proposition}


\begin{corollary}\label{cor: null distribution}
Using notation in Proposition \ref{prop: null distribution biased} and under Assumption~\ref{assum: randomization inference}, we have $|\Pi_{1} \cap  \mathcal{T}|\sim \text{Binomial}\ (I, 1/2)$.
\end{corollary}

According to Proposition \ref{prop: null distribution biased}, the bias RA can be tested by comparing the number of ``correctly inferred" cluster membership to a Binomial distribution and calculating appropriate tail probabilities. Both exact $p$-value ($\overline{p}_{\Gamma,\textsf{exact}}$) and that based on normal approximation ($\overline{p}_{\Gamma,\textsf{approx}}$) can be obtained efficiently.

\subsection{Implementation of CBT: Clustering with side information}
We propose two simple algorithms for CBT, one based on a variant of the $2$-means clustering algorithm and the other based on fitting a two-component mixture model. Both algorithms perform the clustering task with the following side information: one and only one unit in each matched pair is the treated unit.

\begin{side_info}[Matched-Set-Structure Side-Information]\label{side-info: match set structure}
Let $\mathcal{I} = \{1,\dots, I\}$ index the matched pairs, then treatment assignment indicators satisfy $Z_{i1} + Z_{i2} = 1,\forall i \in \mathcal{I}$. Write $\mathcal{F}_{\textsf{side},\textsf{match}} = \left\{Z_{i1} + Z_{i2} = 1,\forall i \in \mathcal{I}\right\}$ and $\mathcal{F}_{\textsf{side},\textsf{match}}$ will be referred to as the matched-set-structure side-information.
\end{side_info}

Side-information \ref{side-info: match set structure} imposes the so-called \emph{Cannot-Link} constraints to the conventional $2$-means clustering algorithm when updating the cluster membership at each iteration, and is known as a \emph{constrained 2-means clustering algorithm} in the literature (\citealp{wagstaff2001constrained}). In the Supplementary Material C, we describe in detail how to solve the constrained $2$-means problem, and how to use a machine learning technique called ``metric learning" to update the distance metric at each iteration of the algorithm, by (i) maximizing the distance between dissimilar pairs, that is, units $i1$ and $i2$ within each matched pair $i$, and (ii) enforcing the distance between units in each cluster and the cluster centroid to be small. 

In addition to $K$-means clustering, another popular clustering method in practice is fitting a mixture model. In our application, the observed covariates data naturally arises from a two-component mixture model $\alpha_{\text{treat}}\cdot F_t(\cdot) + (1-\alpha_{\text{treat}})\cdot F_c(\cdot)$, where $F_t(\cdot)$ and $F_c(\cdot)$ represent some parametric family of covariates distribution in the treated and matched control groups, respectively. Moreover, the mixing probability $\alpha_{\text{treat}} = 0.5$ for the matched-pair design. After fitting a mixture model, say using the EM algorithm (\citealp{dempster1977maximum}), one may then compare the relative magnitude of posterior probabilities of belonging to each cluster of two units in each matched pair, and assign cluster membership accordingly.

\section{Quantifying observed covariates' residual confounding in matched samples}
\label{subsec: residual confounding}
\subsection{Introducing residual sensitivity value (RSV)}
The bounding $p$-values $\{\overline{p}_{\Gamma,\textsf{exact}},~\Gamma \geq 1\}$ testing the biased RA give rise to the following measure of residual confounding, termed \emph{residual sensitivity value}, in a matched sample.

\begin{definition}[Residual Sensitivity Value $\underline{\widetilde{\Gamma}}$]
\label{def: residual sensitivity value}
Given matched-pair data, an exact test for Assumption \ref{assum: biased randomization inference}, and a significance level $\alpha$, the residual sensitivity value $\underline{\widetilde{\Gamma}}$ is the smallest $\Gamma \geq 1$ such that the bounding $p$-value $\overline{p}_{\Gamma,\textsf{exact}}$ is not significant, i.e.,
\begin{equation}
    \underline{\widetilde{\Gamma}} := \inf \{\Gamma \geq 1 \mid \overline{p}_{\Gamma,\textsf{exact}} \geq \alpha\}.
\end{equation}
\end{definition}

The residual sensitivity value $\underline{\widetilde{\Gamma}}$ is a measure of residual confounding from observed covariates in the matched-pair data. Any valid test of Assumption \ref{assum: biased randomization inference} can in principle be inverted sequentially to define the RSV, though we focus on the RSVs derived from SS-CPT and CBT in this article. It is an interesting research topic to explore other powerful statistical tests for Assumption \ref{assum: biased randomization inference}.

By definition, if the RA cannot be rejected at the level $\alpha$, then $\overline{p}_{\Gamma = 1,\textsf{exact}} \geq \alpha$ and hence $\underline{\widetilde{\Gamma}} = 1$. The RSV $\underline{\widetilde{\Gamma}}$ serves at least two practical purposes. First, a large $\underline{\widetilde{\Gamma}}$ signals poor balance after statistical matching and empirical researchers may consider improving the current matched comparison; see Section \ref{subsec: improve balance} below for some useful strategies and methods for this purpose. Second, residual confounding from observed covariates needs to be taken into account in the downstream outcome analysis; the randomization-based outcome analysis should relax the RA and report a bounding $p$-value corresponding to taking $\Gamma = \underline{\widetilde{\Gamma}}$ in a biased randomization scheme. We do need to point out again that although $\underline{\widetilde{\Gamma}}$ provides a measure of residual confounding due to imperfect matching on observed covariates, which represents the smallest degree of deviation from the RA, a sensitivity analysis is still much needed to examine the possibility of unmeasured confounding. Researchers should not confuse the proposed RSV with \citet{zhao2018sensitivity}'s sensitivity value, which measures the minimum strength of unmeasured confounding needed to qualitatively alter the outcome analysis.

\subsection{Practical suggestions: What next after diagnostics}
\label{subsec: improve balance}
If covariate imbalance persists in the matched samples, i.e., the residual sensitivity value $\underline{\widetilde{\Gamma}} > 1$, a practical question emerges as how to address the residual confounding. One strategy is to incorporate the residual sensitivity value $\underline{\widetilde{\Gamma}}$ into the primary outcome analysis as discussed in the last section. Alternatively, researchers could employ some useful study design strategies to improve their matched samples. We briefly discuss a few strategies here.

When the control-to-treated ratio is large and only a few categorical covariates are out of balance, one may add an additional layer in the network-flow-based statistical matching algorithm that directly balances the marginal distribution of these ``recalcitrant" covariates. Such a network architecture is known as ``fine-balance" (\citealp{rosenbaum2007minimum}) and resembles similar counterbalancing strategies like the Latin square designs or crossover designs in the experimental design literature. Many variants of this strategy have been developed (\citealp{yang2012optimal, pimentel2015large}) and this strategy has proven successful in many empirical comparative effectiveness studies (see, e.g., \citealp{silber2007does, silber2013characteristics}). If the recalcitrant variable is continuous, one strategy is to fine-balance its quantiles; alternatively, one may minimize the earth mover's distance (EMD) between the marginal distributions of the recalcitrant continuous variables in the treated and matched control groups (\citealp{zhang2021matching}). These design strategies are implemented in the \textsf{R} packages \textsf{bigmatch} (\citealp{yu2020matching}) and \textsf{match2C} (\citealp{zhang2021matching}).

In many circumstances, the covariate overlap between the treated and control groups before matching may be poor and it is not practical to design a matched comparison that utilizes all treated subjects. For instance, this can happen when patients with certain characteristics almost always receive one treatment or the other, thus violating the positivity assumption (\citealp{Westreich2010positivity}). In such an eventuality, one reasonable strategy is to focus on ``a marginal group of [subjects] who may or may not receive the
treatment depending upon circumstances such as availability, preference, or heterogeneous
opinion" (\citealp{rosenbaum2012optimal}). A statistical matching algorithm that achieves this goal is typically referred to as ``optimal subset matching" and an implementation can be found in \citet{rosenbaum2012optimal} and \citet{Rpkg_rcbsubset}. We will illustrate this strategy when examining a real data example in Section \ref{sec: application RHC}.

Sometimes, researchers have access to one or more instrumental variables (IVs) for the treatment under consideration. For instance, in a renowned study of the effect of community college versus four-year college on educational attainment, \citet{rouse1995democratization} used excess distance to the nearest four-year versus community college and excess tuition of the four-year versus community college as instrumental variables for high school students' enrollment into a community college. If there are data on instrumental variables, then researchers may consider an IV-based matched comparison. For instance, in the above example, one may pair a high school graduate who lived nearer to a four-year college compared to a community college to a student who lived farther from a four-year college compared to a community college. Statistical matching in the IV-based analysis ensures that the IV is valid after controlling for observed covariates. There is often superior overlap between the IV-defined treated and control groups compared to the non-IV study design because the IV is often closer to being randomly assigned. Of course, in an IV-based matched analysis, one needs to carefully define the treatment effects of interest; we refer interested readers to \citet{baiocchi2010building,baiocchi2012near}, \citet{heng2019instrumental} and \citet{zhang2020bridging} for related methodological and theoretical development and \citet{lorch2012differential} and \citet{mackay2021association} for empirical studies.

These strategies are not mutually exclusive: one may need to employ one or more at the same time to construct good matched samples and conduct sound statistical inference.


\section{Simulation studies}
\label{sec: simulation}
\subsection{Goal and structure}
\label{subsec: goal and structure}
We have four goals in the simulation section. First, we compared the power of two implementations of SS-CPT and two implementations of CBT. Second, we compared several statistical matching algorithms and investigated which algorithm produced matched samples with minimal residual imbalance as quantified by the RSV. Third, we examined the relationship between RSV and informal balance diagnostics like the SMDs. Lastly, we investigated the relationship between the performance of downstream randomization inference and both formal and informal diagnostics. 

Our simulation can be summarized as a $2 \times 3 \times 3 \times 4$ factorial design with the following factors:
\begin{description}
\item \textbf{Factor 1:} sample size before matching, $n$: $3000$ and $5000$.
\item \textbf{Factor 2:} observed covariates distribution and overlap: 
$\mathbf X \sim \text{Multivariate Normal}\left(\boldsymbol\mu, \boldsymbol \Sigma\right)$, with $\boldsymbol \mu_{10\times 1} = (cZ, 0, \cdots, 0)^{\text{T}}$ and $\boldsymbol\Sigma = \boldsymbol I_{10\times 10}$ with $c = 0.3$, $0.5$, and $0.7$. 
\item \textbf{Factor 3:} statistical matching algorithm, $\mathcal{M}$: 
\begin{enumerate}
    \item $\mathcal{M}_{\textsf{maha}}$: metric-based matching based on the Mahalanobis distance;
    \item $\mathcal{M}_{\textsf{pscore}}$: propensity score matching with no caliper;
    \item $\mathcal{M}_{\textsf{opt}}$: optimal matching within a $0.2$ SD propensity score caliper.
\end{enumerate}
\item \textbf{Factor 4:} testing procedure: (i). $\textsf{CBT}_{\textsf{K-means}}$: CBT based on a constrained 2-means algorithm; (ii). $\textsf{CBT}_{\textsf{GMM}}$: CBT based on a two-component Gaussian mixture model; (iii). $\textsf{SS-CPT}_{\textsf{accuracy}}$: SS-CPT using $g_{ij}(\hat{\mathbf{f}}_{1\rightarrow2}) = h_{ij}(\hat{\mathbf{f}}_{2\rightarrow1}) = \mathbf{1}\{\hat{f}(\mathbf{x}_{ij})>\hat{f}(\mathbf{x}_{ij^{\prime}})\}$ as scores; (iv). $\textsf{SS-CPT}_{\textsf{pscore}}$: SS-CPT using $g_{ij}(\hat{\mathbf{f}}_{1\rightarrow2}) = h_{ij}(\hat{\mathbf{f}}_{2\rightarrow1}) = \hat{f}(\mathbf{x}_{ij})$ as scores.
\end{description}
Factor $1$ through $3$ define the data-generating processes under consideration; see \citet{rubin1979using} and \citet{zhang2021matching} for some motivation for this simulation setup. In particular, Factor $2$ specifies the observed covariates distribution in the treated and control groups. Parameter $c$ controls the amount of overlap in the treated and control groups before matching: $c = 0.5$ is often considered a moderately large bias (\citealp[expression 3.3]{rubin1979using}). We generated $Z \sim \text{Binomial}(1/3)$ so that the treated-to-control ratio in the simulated datasets was approximately $1$ to $2$. 
Factor $3$ defines three statistical matching algorithms under consideration. All matching algorithms were implemented using the \textsf{R} package \textsf{match2C} (\citealp{zhang2021matching, R_pkg_match2C}). A tutorial of the package can be found via \url{https://cran.r-project.org/web/packages/match2C/vignettes/tutorial.html}. Factor $4$ represents four testing procedures to be studied. There are many other possible implementations of SS-CPT and CBT. Combining SS-CPT and CBT with more powerful machine learning methods may further improve their power. We considered these specific implementations because they are familiar to empirical researchers and easy-to-implement.

Lastly, we generated the potential outcomes $R_T$ and $R_C$ for each unit as follows:
\[
R_T = R_C = 0.2 \times \textsf{sign}\{X_1\} \cdot |X_1|^{0.2} + 0.5\sqrt{|X_2|} - X_3 + \epsilon,\qquad \epsilon \sim \text{Normal}(0, 1).
\]
The observed outcome is $R = R_T \cdot Z + R_C \cdot (1 - Z)$. We focused on this outcome-generating process in this section, and considered additional ones in the Supplementary Material D.

\subsection{Measures of success}
\label{subsec: measure success}
For each data-generating process of $(\mathbf{X}, Z)$ and statistical matching algorithm defined by Factor $1$ through $3$, we repeated the simulation $500$ times and computed the proportion of times the RA was rejected by each of the four algorithms described in Factor $4$. These rejection proportions are denoted by $\textsf{Power}_{\Gamma = 1}$. We then calculated the residual sensitivity values $\underline{\widetilde{\Gamma}}$ determined by each of the four testing procedures for each matched sample. The RSVs complement the rejection proportions at $\Gamma = 1$ by quantifying the extent of deviation from the RA and reflecting the power of each procedure at $\Gamma > 1$. The average RSV is denoted as $\textsf{Mean}~\underline{\widetilde{\Gamma}}$.

To facilitate our understanding of the relationship between informal balance diagnostics and formal statistical tests like those developed in this paper, we also recorded and reported the average standardized mean difference (defined as the difference in means of a covariate in the treated and matched control group, divided by the pooled standard deviation before matching) of the first covariate $X_1$, denoted as $\textsf{SMD}_{X_1}$, and the average median SMD, denoted as $\textsf{SMD}_{0.50}$, across all $10$ covariates. According to our data-generating process, the propensity score distribution in the population is a function of $X_1$ and hence the SMD of $X_1$ after matching essentially captures the SMD of the propensity score after matching.

Finally, we linked the balance diagnostics, both formal and informal, to the performance of randomization-based outcome analysis in two ways. First, we constructed a Hodges-Lehmann point estimate using the \textsf{senWilcox} function in the \textsf{R} package \textsf{DOS} under the RA (assuming $\pi_{i1} = \pi_{i2} = 1/2,~\forall i$ as in \cite{rosenbaum2002observational}) and reported its averaged value across $500$ simulations. This average Hodges-Lehmann estimate is denoted as $\textsf{H-L est}$. Second, we reported its root mean squared error, denoted as $\textsf{RMSE}$ across $500$ simulations, as a measurement of the performance of the outcome analysis. In the Supplementary Material D, we further investigated the performance of statistical inference after incorporating the RSVs.

\subsection{Simulation results}
\begin{table}[ht]
\centering
\caption{Simulation results for various sample sizes $n$, overlap parameter $c$, and statistical matching algorithms $\mathcal{M}$. Roman numerals I to IV are shorthands for $\textsf{CBT}_{\textsf{K-means}}$ (\textsf{I}), $\textsf{CBT}_{\textsf{GMM}}$ (\textsf{II}), $\textsf{SS-CPT}_{\textsf{accuracy}}$ (\textsf{III}), and $\textsf{SS-CPT}_{\textsf{pscore}}$ (\textsf{IV}). The highest power at $\Gamma = 1$ and largest mean RSV for each data-generating process and matching algorithm are highlighted.}
\label{tbl: simulation results main article}
\resizebox{0.9\textwidth}{!}{
\begin{tabular}{cccccccccccccc}
  \hline\\[-0.8em]
  & & & &\multicolumn{4}{c}{$\textsf{Power}_{\Gamma = 1}$} &\multicolumn{4}{c}{$\textsf{Mean}~\underline{\widetilde{\Gamma}}$}&\\
 $n$ & $\mathcal{M}$ & $\textsf{SMD}_{X_1}$ & $\textsf{SMD}_{0.50}$ & $\textsf{I}$ & $\textsf{II}$ &$\textsf{III}$ &$\textsf{IV}$& $\textsf{I}$ & $\textsf{II}$ &$\textsf{III}$ &$\textsf{IV}$ & \textsf{H-L est} & \textsf{RMSE} \\ 
  \hline \\ [-0.8em]
  \multicolumn{14}{c}{c = 0.3}\\
  \multirow{3}{*}{3000} 
& $\mathcal{M}_{\textsf{opt}}$ & 0.01 & 0.01 & 0.04 & \textbf{0.14} & 0.01 & 0.00 & 1.00 & \textbf{1.01} & 1.00 & 1.00 & \textbf{0.02} & \textbf{0.05} \\
& $\mathcal{M}_{\textsf{maha}}$& 0.12 & 0.02 & 0.12 & 0.31 & 0.93 & \textbf{0.97} & 1.01 & 1.02 & 1.16 & \textbf{1.31} & 0.03 & 0.06 \\  
& $\mathcal{M}_{\textsf{pscore}}$ & 0.30 & 0.04 & 0.36 & 0.40 & 0.94 & \textbf{0.98} & 1.03 & 1.03 & 1.17 & \textbf{1.31} & 0.05 & 0.08 \\\\ [-0.8em]

  \multirow{3}{*}{5000} 
  & $\mathcal{M}_{\textsf{opt}}$ & 0.01 & 0.01 & 0.05 & \textbf{0.24} & 0.01 & 0.01 & 1.00 & \textbf{1.01} & 1.00 & 1.00 & \textbf{0.02} & \textbf{0.04} \\ 
  &$\mathcal{M}_{\textsf{maha}}$ & 0.11 & 0.01 & 0.16 & 0.45 & 0.99 & \textbf{1.00} & 1.01 & 1.03 & 1.20 & \textbf{1.38} & 0.03 & 0.05 \\ 
  & $\mathcal{M}_{\textsf{pscore}}$ & 0.30 & 0.03 & 0.56 & 0.53 & 1.00 & \textbf{1.00} & 1.05 & 1.05 & 1.23 & \textbf{1.41} & 0.05 & 0.07 \\\\ [-0.8em]

 \multicolumn{14}{c}{c = 0.5}\\
  \multirow{3}{*}{3000} 
  & $\mathcal{M}_{\textsf{opt}}$ & 0.05 & 0.02 & 0.06 & \textbf{0.23} & 0.07 & 0.16 & 1.00 & 1.01 & 1.00 & \textbf{1.02} & \textbf{0.02} & \textbf{0.05} \\
  & $\mathcal{M}_{\textsf{maha}}$ & 0.21 & 0.02 & 0.46 & 0.47 & 1.00 & \textbf{1.00} & 1.03 & 1.06 & 1.48 & \textbf{1.93} & 0.05 & 0.07 \\  
  & $\mathcal{M}_{\textsf{pscore}}$ & 0.50 & 0.04 & 0.82 & 0.70 & 1.00 & \textbf{1.00} & 1.18 & 1.14 & 1.50 & \textbf{1.92} & 0.08 & 0.10 \\\\ [-0.8em]
  
  \multirow{3}{*}{5000} 
  & $\mathcal{M}_{\textsf{opt}}$ & 0.05 & 0.01 & 0.05 & 0.27 & 0.16 & 0.31 & \textbf{1.00} & 1.01 & 1.01 & \textbf{1.03} & \textbf{0.02} & \textbf{0.04} \\  
  &$\mathcal{M}_{\textsf{maha}}$ & 0.19 & 0.01 & 0.53 & 0.56 & 1.00 & \textbf{1.00} & 1.04 & 1.06 & 1.53 & \textbf{2.02} & 0.04 & 0.06 \\ 
  & $\mathcal{M}_{\textsf{pscore}}$  & 0.50 & 0.03 & 0.92 & 0.83 & 1.00 & \textbf{1.00} & 1.24 & 1.20 & 1.56 & \textbf{2.03} & 0.08 & 0.10 \\
    \\ [-0.8em]

  \multicolumn{14}{c}{c = 0.7}\\
  \multirow{3}{*}{3000} 
  & $\mathcal{M}_{\textsf{opt}}$ & 0.14 & 0.02 & 0.06 & 0.27 & 0.73 & \textbf{0.94} & 1.00 & 1.02 & 1.10 & \textbf{1.54} & \textbf{0.02} & \textbf{0.06} \\ 
  & $\mathcal{M}_{\textsf{maha}}$& 0.31 & 0.02 & 0.90 & 0.64 & 1.00 & \textbf{1.00} & 1.14 & 1.12 & 1.93 & \textbf{2.90} & 0.06 & 0.08 \\ 
  & $\mathcal{M}_{\textsf{pscore}}$ & 0.70 & 0.03 & 0.99 & 0.89 & 1.00 & \textbf{1.00} & 1.49 & 1.42 & 1.88 & \textbf{2.69} & 0.11 & 0.13 \\
    \\ [-0.8em]
  
  \multirow{3}{*}{5000} 
  & $\mathcal{M}_{\textsf{opt}}$ & 0.13 & 0.01 & 0.04 & 0.29 & 0.94 & \textbf{1.00} & 1.00 & 1.02 & 1.20 & \textbf{1.88} & \textbf{0.02} & \textbf{0.04} \\ 
  &$\mathcal{M}_{\textsf{maha}}$ & 0.29 & 0.01 & 0.95 & 0.66 & 1.00 & \textbf{1.00} & 1.16 & 1.11 & 1.99 & \textbf{3.04} & 0.05 & 0.06 \\ 
   & $\mathcal{M}_{\textsf{pscore}}$ & 0.70 & 0.03 & 0.99 & 0.95 & 1.00 & \textbf{1.00} & 1.56 & 1.60 & 1.96 & \textbf{2.86} & 0.11 & 0.12 \\
   \\ [-0.8em]
   \hline
  
\end{tabular}}
\end{table}

Table \ref{tbl: simulation results main article} summarizes the simulation results for different choices of sample size $n$, overlap parameter $c$, statistical matching algorithm $\mathcal{M}$, and testing procedure. We identified several consistent trends from simulation results. First, each of the four testing procedures had improved power when testing the RA ($\textsf{Power}_{\Gamma = 1}$) and identified a larger RSV ($\textsf{Mean}~\underline{\widetilde{\Gamma}}$) for larger sample size $n$ and worse matching quality as captured by $\textsf{SMD}_{X_1}$. The testing procedure $\textsf{SS-CPT}_{\textsf{pscore}}$ (Method \textsf{IV} in Table \ref{tbl: simulation results main article}) seemed to have the largest power when testing the RA and identify the largest RSV in $15/18$ simulation settings. Interestingly, $\textsf{SS-CPT}_{\textsf{pscore}}$ appeared to be superior to $\textsf{SS-CPT}_{\textsf{accuracy}}$ in testing the RA and identifying the RSV in almost all simulation settings. The testing procedure $\textsf{CBT}_{\textsf{GMM}}$ appeared to be slightly more favorable in the other $3$ simulation settings where the sample was well matched and $\textsf{SMD}_{X_1}$ was small. Based on the simulation results, we would recommend using $\textsf{CBT}_{\textsf{GMM}}$ when the largest SMD is less than $0.05$, or one-twentieth of one pooled standard deviation, and $\textsf{SS-CPT}_{\textsf{pscore}}$ otherwise. 

Second, upon examining Table \ref{tbl: simulation results main article}, we found that optimal matching $\mathcal{M}_{\textsf{opt}}$ delivered the best matched samples in the following senses: (i) the average SMD of $X_1$ which served as an informal measure of overall balance was much smaller for $\mathcal{M}_{\textsf{opt}}$ compared to $\mathcal{M}_{\textsf{pscore}}$ or $\mathcal{M}_{\textsf{maha}}$; (ii) with even a moderately large bias before matching ($c = 0.5$) and a large sample size ($n = 5000$), testing procedures could barely reject the RA on an optimally matched sample (highest $\textsf{Power}_{\Gamma = 1} = 0.31$ for $\textsf{SS-CPT}_{\textsf{pscore}}$), and identify little deviation (largest mean RSV = $1.03$ for $\textsf{SS-CPT}_{\textsf{pscore}}$).

Third, we concluded that RSVs supplemented informal diagnostics like SMDs. In empirical studies, an informal rule of thumb says that SMDs of all observed covariates should be less than $0.1$ (\citealp{silber2001multivariate,austin2015moving}). In our simulation studies, many matched samples satisfied this rule of thumb, e.g., optimal matching when $c = 0.3$ or $0.5$ and Mahalanobis-distance-based matching when $c = 0.3$; however, even in these circumstances, the RA was frequently rejected. For instance, when $c = 0.3$ and $\mathcal{M} = \mathcal{M}_{\textsf{maha}}$, $\textsf{SS-CPT}_{\textsf{pscore}}$ rejected $485/500$ matched datasets and identified an RSV as large as $1.31$ on average. The bottom line is that the RSVs provide a formal and quantitative assessment of the quality of matched sample.

Lastly, what are the implications of statistical matching quality on the randomization-based outcome analysis? Table \ref{tbl: simulation results main article} suggests that the root mean squared error (RMSE) of the randomization-based outcome analysis of optimally-matched samples ($\mathcal{M} = \mathcal{M}_{\textsf{opt}}$) seemed to be the best among three matching methods across \emph{all} data-generating processes. Figure \ref{fig: boxplot plots of HL} further presents the boxplots of the Hodges-Lehmann point estimate in matched samples when the sample size $n = 3000$ and the RSV as determined by $\textsf{SS-CPT}_{\textsf{pscore}}$ equal to $1$ (i.e., the RA is not rejected) versus those with RSV larger than $1$ (i.e., the RA is rejected). We found that the estimate was less biased when the RA was not rejected by $\textsf{SS-CPT}_{\textsf{pscore}}$. We also found that the confidence intervals corresponding to the partial identification intervals under a biased RA scheme with $\Gamma = \underline{\widetilde{\Gamma}}$ almost always obtained nominal level, although very conservative, due to the conservative nature of partial identification and Rosenbaum bounds; see Supplementary Material D for details.

We repeated a subset of simulation studies with multiple random initializations when running CBT and multiple random sample-splitting schemes when running SS-CPT, and found that $\textsf{Power}_{\Gamma = 1}$ and $\textsf{Mean}~\underline{\widetilde{\Gamma}}$ remained stable.

\begin{figure}
    \centering
    \includegraphics[width=\textwidth]{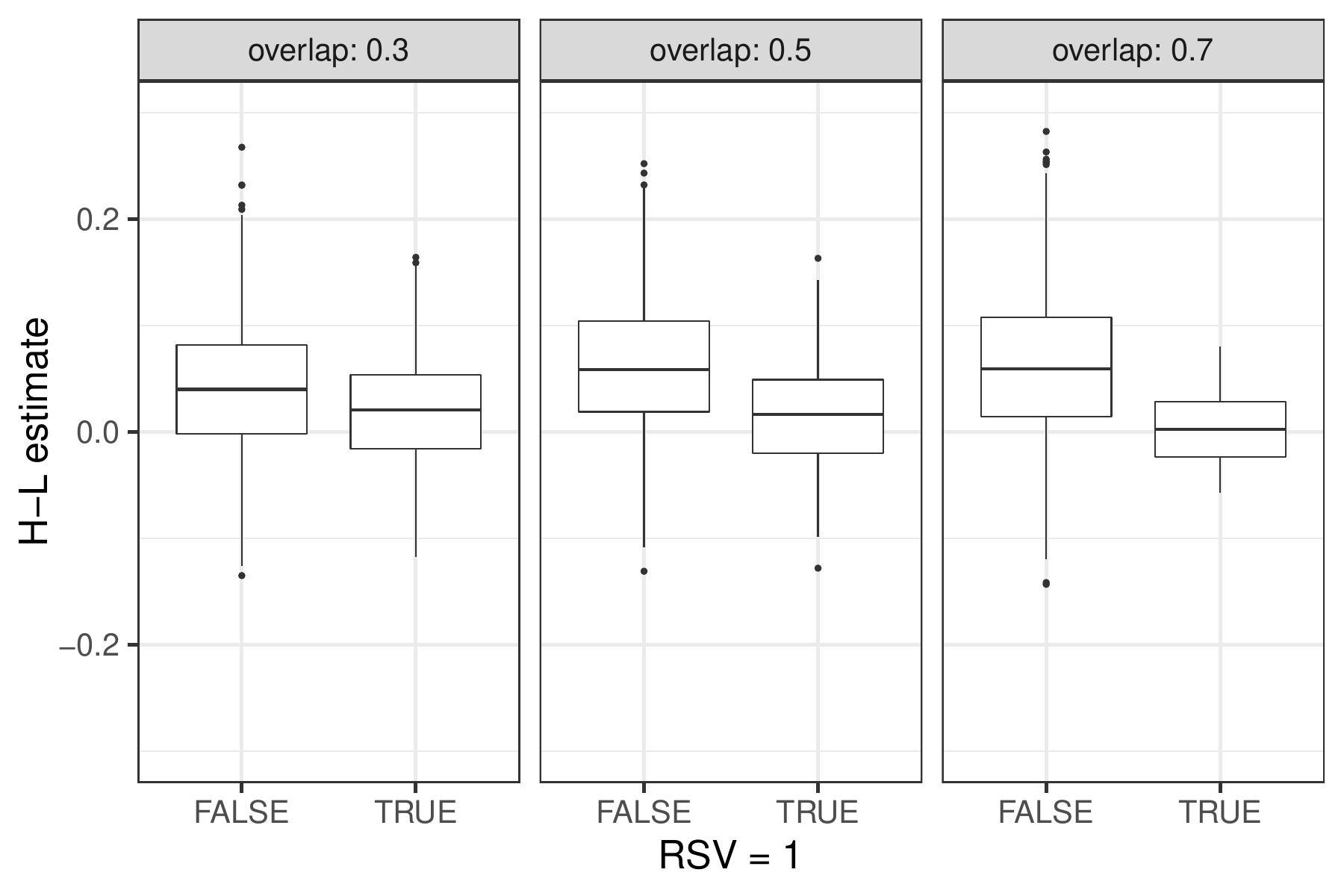}
    \caption{Boxplots of Hodges-Lehmann point estimate in the matched samples with $n = 3000$ and RSV determined by $\textsf{SS-CPT}_{\textsf{pscore}}$ equal to $1$ versus those with larger than $1$.}
    \label{fig: boxplot plots of HL}
\end{figure}
\section{Case study: Effectiveness of right heart catheterization (RHC) in the initial care of critically ill patients}
\label{sec: application RHC}
\subsection{Background, all samples and matched samples}
We illustrate various aspects of our proposed framework by applying it to an observational study investigating the effectiveness of right heart catheterization (RHC) in the initial care of critically ill patients. Since its introduction into the intensive care units (ICUs) almost $50$ years ago, RHC was perceived by clinicians as largely beneficial because direct measurements of heart function by RHC help guide therapies, which was believed to lead eventually to more favorable patient outcomes (\citealp{connors1983evaluation,connors1996effectiveness}). In fact, the belief in RHC's effectiveness was so strong that many physicians refused to enroll their patients into randomized controlled trials for ethical considerations (\citealp{guyatt1991randomized}). In the absence of RCTs, \citet{connors1996effectiveness} examined the effectiveness of RHC in a matching-based prospective cohort study using observational data. We followed \citet{rosenbaum2012optimal} and considered patients under the age of $65$ from the work of \citet{connors1996effectiveness}; among these patients there are $1194$ patients who received RHC and $1804$ who did not. We followed \citet{connors1996effectiveness} and considered observed covariates related to patients demographics, laboratory measurements and vital signs because these factors clearly affect both physicians' decision to use or withhold RHC and patients' outcomes. The outcome of interest is patients' 30-day mortality.

\begin{table}[ht]
\centering
\caption{\small Balance table of all samples and the matched samples. There are a total of $1194$ patients under $65$ undergoing RHC and $1804$ not undergoing RHC. A total of $1194$ matched pairs were formed in the optimal match \textsf{M1} and $1061$ pairs were formed in the optimal subset match \textsf{M2}. SMD is the abbreviation of the standardized mean difference which is equal to the difference in means of covariates in the treated and matched control groups divided by the pooled standard deviation in the treated and control groups before matching.}
\label{tbl: balance table}
\resizebox{\textwidth}{!}{
\begin{tabular}{lcccccccc}
  \hline
 & \multirow{3}{*}{\begin{tabular}{c}All \\RHC \\ ($n = 1194$)\end{tabular}}  & \multirow{3}{*}{\begin{tabular}{c}All\\ No RHC\\ ($n = 1804$)\end{tabular}} & \multirow{3}{*}{\begin{tabular}{c}Abs. \\ SMD\end{tabular}} & \multirow{3}{*}{\begin{tabular}{c}Optimal\\ No RHC \\ ($n = 1194$)\end{tabular}} & \multirow{3}{*}{\begin{tabular}{c}Abs. \\ SMD\end{tabular}} &
 \multirow{3}{*}{\begin{tabular}{c}Subset\\ RHC \\ ($n = 1061$)\end{tabular}} &
 \multirow{3}{*}{\begin{tabular}{c}Subset\\ No RHC \\ ($n = 1061$)\end{tabular}} &
 \multirow{3}{*}{\begin{tabular}{c}Abs. \\ SMD\end{tabular}}\\ \\ \\
  \hline
  \textbf{Demographics} \\
  \hspace{0.3cm} Age, yr & 49.56 & 48.01 & 0.09 & 49.26 & 0.02 & 48.65 & 48.07 & 0.03\\ 
  \hspace{0.3cm} Male, \% & 0.58 & 0.57 & 0.02 & 0.59 & 0.01 & 0.57 & 0.58 & 0.02\\ 
  \hspace{0.3cm} White, \% & 0.72 & 0.72 & 0.01 & 0.73 & 0.02 & 0.73 & 0.72 & 0.01\\ 
  \textbf{Laboratory Measurements} \\
  \hspace{0.3cm} PaO$_2$/FiO$_2$, mm Hg & 196.13 & 243.68 & 0.29 & 214.55 & 0.11 &225.81 & 226.36 & 0.00\\ 
  \hspace{0.3cm} PaO$_2$, mm Hg & 36.74 & 38.67 & 0.11 & 37.62 & 0.05 &37.55 & 37.38 & 0.01\\ 
  \hspace{0.3cm} WBC count, $\times 10^9$L & 15.77 & 14.98 & 0.05 & 15.29 & 0.03 & 15.84 & 15.30 & 0.03\\ 
  \hspace{0.3cm} Creatinine, $\mu$mol/L & 2.46 & 1.94 & 0.17 & 2.15 & 0.10 & 2.07 & 2.16 & 0.03\\ 
  \textbf{Vital Signs} \\
  \hspace{0.3cm} Blood pressure, mm Hg & 69.44 & 87.47 & 0.35 & 75.95 & 0.13 & 80.07 & 79.86 & 0.00\\ 
  \hspace{0.3cm} APACHE III score & 61.07 & 50.51 & 0.36 & 55.63 & 0.19 & 53.34 & 54.56 & 0.04\\ 
  \hspace{0.3cm} Coma score & 17.90 & 21.66 & 0.09 & 16.49 & 0.03 & 20.40 & 20.03 & 0.01\\ 
  \textbf{Propensity Score} & 0.48 & 0.35 & 0.55 & 0.42 & 0.24 & 0.40 & 0.40 & 0.00 \\ 
   \hline
\end{tabular}}
\end{table}

The first three columns of Table \ref{tbl: balance table} summarize patients' characteristics in the full sample. Not surprisingly, a number of important variables are vastly different between the two groups; for instance, the mean blood pressure is only $69.44$ mm Hg in the RHC group compared to $87.47$ mm Hg in the no RHC group. Overall, the two groups before matching have similar demographics but the RHC group was sicker. We considered two matched samples, the first being an optimal pair match that utilized all $n = 1194$ RHC patients and formed $1194$ matched pairs of two patients (one receiving RHC and the other not), and the second being an optimal subset match that formed $1061$ matched pairs (\citealp{rosenbaum2012optimal}). The optimal subset match was performed using the \textsf{R} package \textsf{rcbsubset} with default settings (\citealp{Rpkg_rcbsubset}). We will refer to the optimal match design as $\textsf{M1}$ and the optimal subset match as $\textsf{M2}$. It is evident that both matched designs helped remove overt bias in the observed covariates; however, the question remains as to which design, if any of them, could justify a randomization-based analysis of the primary outcome of 30-day mortality status.

\subsection{Quantifying residual confounding and conducting outcome analysis}
Upon applying $\textsf{SS-CPT}_{\textsf{accuracy}}$ and $\textsf{SS-CPT}_{\textsf{pscore}}$ to the design $\textsf{M1}$, we obtained an RSV of $2.52$ and $5.67$, respectively. The $\textsf{CBT}$ based on fitting a two-component Gaussian mixture model yielded an RSV equal to $1.21$. All three tests rejected the RA for $\textsf{M1}$, though two $\textsf{SS-CPT}$ implementations were more powerful, which seemed to agree well with our simulation results in similar settings. Next, we performed a randomization-based outcome analysis using McNemar's test against the alternative hypothesis that RHC has a negative effect on 30-day mortality. Under the randomization assumption, the $p$-value of outcome analysis was $0.022$, which provided some weak evidence against the null hypothesis of no effect. However, this treatment effect immediately disappeared when taking into account the RSVs obtained using either SS-CPT or CBT. In other words, the observed treatment effect under the RA was likely to be an artifact of imperfect matching and residual imbalance in observed covariates, and a researcher who did not fully take this into account could draw a false conclusion. 

On the other hand, none of our proposed tests rejected the RA for the design $\textsf{M2}$. Again, this aligns well with our simulation results as all absolute SMDs are less than $0.05$ in $\textsf{M2}$, and we found in simulation studies that our testing procedures had little power in similar settings. For the matched design $\textsf{M2}$, we conducted outcome analysis under the RA, and obtained a $p$-value equal to $0.79$, suggesting insufficient evidence against the causal null hypothesis for the study units involved in the design \textsf{M2}. All outcome analyses (under $\Gamma = 1$ or $\Gamma = \underline{\widetilde{\Gamma}}$) were performed using \textsf{R} package \textsf{rbounds} (\citealp{R_pkg_rbounds}). 

In this example, there is a tension between the internal validity (i.e., comparability of the RHC and no RHC groups) and external validity (i.e., how results generalize to a target population): the design \textsf{M1} has superior generalizability over \textsf{M2} as \textsf{M1} utilizes the entire treated group, while the design \textsf{M2} has far superior internal validity as it better approximates an ideal RCT; see \citet{Zhang_template_match} for a method that builds a well-matched sample mimicking a target population.

\section{Discussion: Summary, strengths and weaknesses of proposed tests, and extensions}
\label{sec: discussion}
 It is a popular strategy in comparative effectiveness research to embed observational data into a randomized controlled experiment using statistical matching and analyze matched data as if it were a randomized experiment. As Collin Mallows famously pointed out (see, e.g., \citealp{denby2013conversation}), the most robust statistical technique
is to look at the data; a matched observational study is therefore robust in the sense that it forces researchers to examine the covariate balance and overlap, focus on the covariate space that is well-overlapped, and avoid unfounded extrapolation (\citealp{ho2007matching,rubin2007design,stuart2010matching,rosenbaum2002observational,rosenbaum2010design}). Despite a preferred strategy to draw causal conclusion (in our opinion), there is a gap between an approximate experiment (i.e., data after statistical matching) and a genuine experiment, and this gap is often circumvented by making the randomization assumption justified by informal or formal balance diagnostics.

One important limitation of statistical tests developed for the randomization assumption is that these tests cannot quantify the extent to which the randomization assumption is violated due to residual imbalance in observed covariates. Our proposed testing framework is thus advantageous in its ability to quantify the deviation from the randomization assumption using the residual sensitivity value $\underline{\widetilde{\Gamma}}$. Although our primary focus in the article is matched-pair design, the framework and algorithm can be readily extended to matching with multiple controls; see Supplementary Material B for analogous results for matching-with-multiple-controls designs.

Both SS-CPT and CBT can be used in combination with any user-chosen classification or constraint clustering algorithms. Our simulation studies compared only two specific implementations of each method; more powerful classification and constraint clustering methods could potentially deliver more powerful statistical tests. The SS-CPT method with a propensity score defined score function appeared to be the most powerful in most settings, while the CBT method based on Gaussian mixture modeling appeared to be slightly more advantageous in very closely matched sample. Compared to that of the CBT method, calculation of the exact bounding $p$-value of the SS-CPT method is less computationally efficient for a generic score function.

We recommend empirical researchers to examine the covariate balance using both formal and informal diagnostics, and when possible, incorporate the level of residual confounding into their outcome analysis. For instance, one way to do this is to perform a randomization inference under a biased randomization scheme using Rosenbaum bounds (\citealp{rosenbaum2002observational, rosenbaum2010design}) with the parameter $\Gamma$ set to the magnitude of the residual sensitivity value, i.e., $\Gamma = \underline{\widetilde{\Gamma}}$. Other strategies to formally reflect the study design quality in the downstream outcome analysis are worth exploring.

Failure to reject our proposed test, like failure to reject any statistical test, does not translate to a statement about the correctness of the randomization assumption; in fact, statistical matching algorithms are likely to create dependence among matched pairs or sets so that the independence part of the randomization assumption almost surely does not hold. However, through our extensive simulations (see also simulations in \citet{branson2018randomization}), it appears that when the randomization assumption cannot be rejected by our proposed tests, the randomization-based outcome analysis typically has good statistical performance; in other words, the randomization assumption is a good approximation of the complicated reality in these cases.

Lastly, in order to draw high-quality, convincing causal conclusions, one necessarily needs to perform extensive sensitivity analysis that allows for some hypothetical unmeasured confounding. To stress, while our proposed residual sensitivity value takes into account the deviation from the randomization assumption due to residual observed covariates imbalance, it says nothing about unmeasured confounding; in fact, unmeasured confounding can still bias the random assignment in each matched pair to an arbitrary extent even when the residual sensitivity value is $1$. We recommend reporting the outcome analysis with both a residual sensitivity value and \citet{zhao2018sensitivity}'s sensitivity value that examines the maximum extent of deviation from randomization (possibly due to unmeasured confounding) needed to qualitatively alter the causal conclusion are reported.

\section*{Supplementary Materials}
Extension of the proposed methodology to matching-with-multiple-controls, proofs, additional simulation results, and code.

\section*{Acknowledge}
We would like to acknowledge the editor, associate editor, and three anonymous reviewers for their careful reviews and constructive comments which largely improved the article.

\bibliographystyle{apalike}
\bibliography{paper-ref}

\clearpage
\pagenumbering{arabic}
\begin{center}
    \Large Supplementary Materials for ``Testing relaxed randomization assumptions and quantifying imperfect matching and residual confounding in matched observational studies"
\end{center}

\section*{Supplementary Material A: Detailed Literature Review}
Researchers routinely examine the covariate balance based on certain metrics. Perhaps the most widely used metric is standardized mean differences of all observed covariates before and after matching. A rule of thumb widely appreciated and accepted in empirical comparative effectiveness research is that the absolute standardized mean differences (SMDs) of all covariates are less than $0.1$, or one-tenth of one pooled standard deviation (\citealp{silber2001multivariate}, \citealp{austin2015moving}). More recently, \citet{franklin2014metrics} proposed to use c-statistic as an alternative to the SMDs to examine covariate balance.

Some authors have proposed formal statistical procedures that test the equality of multivariate covariate distributions in the treated and matched control groups (see, e.g., \citealp{hansen2008covariate, austin2009balance, heller2010using,chen2016new,yu2020evaluating}, among others). One prominent example is the cross-match test (\citealp{rosenbaum2005exact,heller2010using}). The cross-match test pairs study units using a statistical matching technique called optimal non-bipartite matching (\citealp{lu2011optimal,baiocchi2012near, zhang2022statistical}), takes as the test statistic the number of times a treated unit is paired with a control unit, and rejects the null hypothesis of equal distribution when the statistic is small.

Another closely-related topic, albeit underexplored in the observational studies literature, is hypothesis testing of the underlying law generating data $X$ belonging to a certain family of distributions based on the \emph{group invariance} structure (\citealp{lehmann1949theory,romano1990behavior,kuchibhotla2020exchangeability,dobriban2021consistency}): for the matched pair data, the equality of covariate distributions in the treated and control groups can be recast as the invariance of the probability distribution of the data under a transformation group that flips the treatment status within each matched pair. One can then construct a standard randomization-based, exact test based on the invariance structure. These aforementioned tests target the null hypothesis of the equality of covariate distributions; however, there is a gap between equality of covariate distributions in the treated and matched control groups and the RA. The downstream, Fisherian \emph{finite-sample} randomization inference relies solely on the RA as the treatment assignment is the only source of randomness in the statistical inference and the procedure need not assume sampling covariates or potential outcomes from some superpopulation. In other words, tests that do not target the RA are best understood as an aid to appraise where a matched comparison stands in relation to a recognizable benchmark (\citealp{heller2010using}).

\section*{Supplementary Material B: Proofs and Additional Results for Matching with Multiple Controls}
\subsection*{B.1: Proof of Proposition~\ref{prop: null distribution biased for CPT}}
\begin{proof}[Proof of Proposition~\ref{prop: null distribution biased for CPT}]
Note that for each $i$, under Assumption~\ref{assum: biased randomization inference}, we have $\sum_{j=1}^{2}Z_{ij}q_{ij}$ is stochastically dominated by $\overline{T}_{\Gamma i}$, i.e., for any $t$, we have $P\left(\sum_{j=1}^{2}Z_{ij}q_{ij} \geq t \mid \mathbf{X}, \textsf{M}\right)\leq P\Big(\overline{T}_{\Gamma i} \geq t \Big)$. Then Proposition~\ref{prop: null distribution biased for CPT} follows immediately from the independence of matched pairs $i=1,\dots, I$ and the fact that $T=\sum_{i=1}^{I}\sum_{j=1}^{2}Z_{ij}q_{ij}$.
\end{proof}

\subsection*{B.2: Proof of Proposition~\ref{prop: null distribution biased}}
\begin{lemma}[Bounding one-sided $p$-value]\label{lemma: one sided null distribution biased}
Let $\mathcal{T} = \{ij: Z_{ij}=1, i=1,\dots, I, j=1, 2\}$ be the set of indices corresponding to the treated units in each matched pair. Let $\Pi_1$ and $\Pi_2$ be the output from an appropriate algorithm $\textsf{ALG}$ in Definition \ref{def: appropriate algorithm}. For $i=1,\dots,I$, define $\widetilde{T}_{\Gamma i}$ to be independent random variables taking the value $1$ with probability $\Gamma/(1+\Gamma)$ and the value $0$ with probability $1/(1+\Gamma)$. Under Assumption~\ref{assum: biased randomization inference} with $\Gamma \geq 1$, we have for any $t$,
\begin{align*}
P\left(|\Pi_{1} \cap  \mathcal{T}| \geq t \mid \mathbf{X}, \textsf{M}\right)&\leq P\Big(\sum_{i=1}^{I}\widetilde{T}_{\Gamma i} \geq t \Big)\simeq 1-\Phi\left(\frac{t-\left\{\Gamma/(1 + \Gamma)\right\}\cdot I }{ \left\{\Gamma/(1+\Gamma)^2 \cdot I\right\}^{1/2}}\right), \\
 P\left(|\Pi_{1} \cap  \mathcal{T}| \leq t \mid \mathbf{X}, \textsf{M}\right)&\leq P\Big(\sum_{i=1}^{I}(1-\widetilde{T}_{\Gamma i}) \leq t\Big) \simeq \Phi\left(\frac{t-\left\{1/(1 + \Gamma)\right\}\cdot I }{ \left\{\Gamma/(1+\Gamma)^2 \cdot I\right\}^{1/2}}\right),
\end{align*}
where $\Phi(\cdot )$ is the distribution function of standard normal distribution and ``$\simeq$" denotes that two sequences are asymptotically equal as $I \rightarrow \infty$.
\end{lemma}

\begin{proof}[Proof of Lemma~\ref{lemma: one sided null distribution biased}]
By Definition~\ref{def: appropriate algorithm} and Assumption~\ref{assum: biased randomization inference}, we have $\mathbf{1}\{Z_{ij_{1}}=1\}$, $i=1,\dots,I$ are i.i.d. random variables such that $P(\mathbf{1}\{Z_{ij_{1}}=1\}=1\mid \mathbf{X}, \textsf{M})=P(Z_{ij_{1}}=1\mid \mathbf{X}, \textsf{M})\in [\frac{1}{1+\Gamma}, \frac{\Gamma}{1+\Gamma}]$. Therefore, we have 
\begin{align*}
    P\left(|\Pi_{1} \cap  \mathcal{T}| \geq t \mid \mathbf{X}, \textsf{M}\right)&=P\left(\sum_{i=1}^{I}\mathbf{1}\{Z_{ij_{1}}=1\} \geq t \mid \mathbf{X}, \textsf{M}\right) \\
    &\leq P\Big(\sum_{i=1}^{I}\widetilde{T}_{\Gamma i} \geq t \Big)\simeq 1-\Phi\left(\frac{t-\left\{\Gamma/(1 + \Gamma)\right\}\cdot I }{ \left\{\Gamma/(1+\Gamma)^2 \cdot I\right\}^{1/2}}\right)
\end{align*}
and 
\begin{align*}
    P\left(|\Pi_{1} \cap  \mathcal{T}| \leq t \mid \mathbf{X}, \textsf{M}\right)&=P\left(\sum_{i=1}^{I}\mathbf{1}\{Z_{ij_{1}}=1\} \leq t \mid \mathbf{X}, \textsf{M}\right)\\
    &\leq P\Big(\sum_{i=1}^{I}(1-\widetilde{T}_{\Gamma i}) \leq t\Big)\simeq \Phi\left(\frac{t-\left\{1/(1 + \Gamma)\right\}\cdot I }{ \left\{\Gamma/(1+\Gamma)^2 \cdot I\right\}^{1/2}}\right).
\end{align*}
So the desired conclusion follows.
\end{proof}

\begin{proof}[Proof of Proposition~\ref{prop: null distribution biased}]
We have for any $t\in \mathbb{R}$,
\begin{align*}
 &\quad P\left((|\Pi_{1} \cap  \mathcal{T}|-I/2)^{2} \geq (t-I/2)^{2} \mid \mathbf{X}, \textsf{M}\right)\\
 &=P\left(|\Pi_{1} \cap  \mathcal{T}| \geq |t-I/2|+I/2 \mid \mathbf{X}, \textsf{M}\right)+P\left(|\Pi_{1} \cap  \mathcal{T}| \leq -|t-I/2|+I/2 \mid \mathbf{X}, \textsf{M}\right)\\
 &\leq  P\Big(\sum_{i=1}^{I}\widetilde{T}_{\Gamma i} \geq |t-I/2|+I/2 \Big)+P\Big(\sum_{i=1}^{I}(1-\widetilde{T}_{\Gamma i}) \leq -|t-I/2|+I/2 \Big) \qquad \qquad \text{(By Lemma~\ref{lemma: one sided null distribution biased})} \\
 &\simeq 1-\Phi\left(\frac{|t-I/2|+I/2-\left\{\Gamma/(1 + \Gamma)\right\}\cdot I }{ \left\{\Gamma/(1+\Gamma)^2 \cdot I\right\}^{1/2}}\right) + \Phi\left(\frac{-|t-I/2|+I/2-\left\{1/(1 + \Gamma)\right\}\cdot I }{ \left\{\Gamma/(1+\Gamma)^2 \cdot I\right\}^{1/2}}\right).
\end{align*}
So the desired conclusion follows. 
\end{proof}

\subsection*{B.3: Extension to $1$-to-$k$ match}
\subsubsection*{B.3.1: Setup}
To extend the result to matching with multiple controls, we assume in each matched set, there are $K+1$ units, and $n = I \times (K+1)$. We first state the randomization assumption and the biased randomization assumption when matching with multiple controls. 

\begin{assumption}[Randomization Assumption in Matched Set Studies, Stated Formally]\label{assum: randomization inference multiple}
The treatment assignments across matched sets are assumed to be independent of each other, with
\begin{equation}
\label{eqn: assumption 1 multiple controls}
P(Z_{ij}=1 \mid \mathbf{X}, \sum_{j'=1}^{K+1} Z_{ij'}=1)= \frac{1}{K+1},   
\end{equation}
for $i = 1, 2, \dots, I$ and $j=1,\dots, K+1$.
\end{assumption}

\begin{assumption}[Biased Randomization Assumption in Multiple Controls Case]\label{assum: biased randomization inference: multiple controls}
Suppose that there are $I$ matched sets, and there is one treated and $K$ controls in each matched set. We assume that treatment assignments across matched sets are independent, with
\begin{equation*}
\Gamma^{-1} \leq \frac{ P(Z_{ij}=1\mid \mathbf{X}, \sum_{j=1}^{K+1}Z_{ij}=1)}{P(Z_{ij^{\prime}}=1\mid \mathbf{X}, \sum_{j=1}^{K+1}Z_{ij}=1)} \leq \Gamma
\end{equation*}
for all $i = 1, 2, \dots, I$ and $j, j^{\prime}\in \{1,\dots,K+1\}$, and some $\Gamma \in [1, \infty)$.
\end{assumption}

\subsubsection*{B.3.2: Extending SS-CPT to 1-to-k match}

We show how to extend SS-CPT to the setting in which one treated subject can be matched with k controls ($k \geq 2$), of which the detailed procedure will be summarized in Algorithm~\ref{alg: SS-CPT with 1-to-k}. We first state the following result from \citet{gastwirth2000asymptotic} for calculating the worst-case p-value under Assumption~\ref{assum: biased randomization inference: multiple controls} for a general test statistics $T=\sum_{i=1}^{I}\sum_{j=1}^{k}Z_{ij}q_{ij}$, where each $q_{ij}$ is some fixed score.

\begin{proposition}[Asymptotic separability algorithm in \citet{gastwirth2000asymptotic}]\label{prop: asymptotic separability}

Consider a test statistic $T=\sum_{i=1}^{I}\sum_{j=1}^{k}Z_{ij}q_{ij}$, where each $q_{ij}$ is some fixed score. Let $b\in \{1,\dots,k-1\}=:[k-1]$, and define
\begin{equation*}
	\overline{\overline{\mu}}_{ib}=\frac{\sum_{j=1}^{b} q_{i(j)}+\Gamma \sum_{j=b+1}^{k} q_{i(j)}}{b+\Gamma(k-b)}, \quad i=1,\dots,I, \ b=1, \dots, k-1,
\end{equation*}
\begin{equation*}
	\overline{\overline{\nu}}_{ib}=\frac{\sum_{j=1}^{b} q^{2}_{i(j)}+\Gamma \sum_{j=b+1}^{k} q^{2}_{i(j)}}{b+\Gamma(k-b)}-\overline{\overline{\mu}}_{ib}^{2}, \quad i=1,\dots,I, \ b=1, \dots, k-1,
\end{equation*}
where we rearrange $q_{i1}, \dots, q_{ik}$ as $q_{i(1)} \leq \dots \leq q_{i(k)}$. Let $\overline{\overline{\mu}}_{i}=\max_{b\in [k-1]}\overline{\overline{\mu}}_{ib}$, $B_{i}=\{b: \overline{\overline{\mu}}_{ib}=\overline{\overline{\mu}}_{i}, b\in [k-1] \}$ and $\overline{\overline{\nu}}_{i}=\max_{b \in B_{i}}\overline{\overline{\nu}}_{ib}$. Under mild regularity conditions specified in \citet{gastwirth2000asymptotic}, given any observed value $t$ of the test statistic $T$, the one-sided worst-case p-value under Assumption~\ref{assum: biased randomization inference: multiple controls} can be asymptotically approximated via
\begin{equation*}
\max	P\left(T \geq t \mid \mathbf{X}, \textsf{M}\right) \simeq 1-\Phi \Bigg(\frac{t-\sum_{i=1}^{I} \overline{\overline{\mu}}_{i}}{\sqrt{\sum_{i=1}^{I}\overline{\overline{\nu}}_{i}}} \Bigg) \quad \text{as $I\rightarrow \infty$}.
\end{equation*}
\end{proposition}

Then we are ready to present the following algorithm for testing Assumption~\ref{assum: biased randomization inference: multiple controls} using SS-CPT with 1-to-k match.

\smallskip

\begin{algorithm}[H]

\caption{Algorithm for Testing Assumption~\ref{assum: biased randomization inference: multiple controls} using SS-CPT with 1-to-k match}
\label{alg: SS-CPT with 1-to-k}

\textbf{Input:} $\mathbf{X}=\{\mathbf{x}_{ij}: i=1,\dots, I, j = 1, \dots, k\}$ and $\mathbf{Z}=\{Z_{ij}: i=1,\dots, I, j = 1, \dots, k\}$. 

\textbf{1.} Randomly split the matched sample into two parts: $\mathcal{I}^{(1)}$ and $\mathcal{I}^{(2)}$ such that $\mathcal{I}^{(1)} \cup \mathcal{I}^{(2)}=\{1,\dots, I\}$ and $\mathcal{I}^{(1)} \cap \mathcal{I}^{(2)}=\emptyset$. 

\textsf{2.} Train a classifier $\hat{f}_{1}$ with covariates $\mathbf{X}^{(1)}=\{ X_{ij}, i \in \mathcal{I}^{(1)}, j=1,\dots, k\}$ as predictors and treatment status $\mathbf{Z}^{(1)}=\{ Z_{ij}, i \in \mathcal{I}^{(1)}, j=1,\dots, k\}$ as labels. 

\textsf{3.} Train a classifier $\hat{f}_{2}$ with covariates $\mathbf{X}^{(2)}=\{ X_{ij}, i \in \mathcal{I}^{(2)}, j=1, \dots, k\}$ as predictors and treatment status $\mathbf{Z}^{(2)}=\{ Z_{ij}, i \in \mathcal{I}^{(2)}, j=1,\dots, k\}$ as labels.

\textsf{4.} Let $\hat{\mathbf{f}}_{1\rightarrow2}=\{\hat{f}_1(\mathbf{x}_{ij}), i\in \mathcal{I}^{(2)}, j=1,\dots, k\}$ denote the predicted scores for $\mathcal{I}^{(2)}$ based on the classifier $\hat{f}_{1}$. Define the test statistic $T_{1\rightarrow 2}=\sum_{i\in \mathcal{I}^{(2)} } \sum_{j=1}^{k}Z_{ij}g_{ij}(\hat{\mathbf{f}}_{1\rightarrow2})$, and calculate the worst-case $p$-value under the biased RA with $\Gamma$ according to Proposition~\ref{prop: asymptotic separability}. Denote this bounding $p$-value by $p_{1\rightarrow 2, \Gamma}$.

\textsf{5.} Let $\hat{\mathbf{f}}_{2\rightarrow1}=\{\hat{f}_2(\mathbf{x}_{ij}), i\in \mathcal{I}^{(1)}, j=1,\dots, k\}$ denote the predicted scores for $\mathcal{I}^{(1)}$ based on the classifier $\hat{f}_{2}$. Define the test statistic $T_{2\rightarrow 1}=\sum_{i\in \mathcal{I}^{(1)} } \sum_{j=1}^{k}Z_{ij}h_{ij}(\hat{\mathbf{f}}_{2\rightarrow 1})$, and calculate the worst-case $p$-value under the biased RA with $\Gamma$ according to Proposition~\ref{prop: asymptotic separability}. Denote this bounding $p$-value by $p_{2\rightarrow 1, \Gamma}$.

\textbf{Output:} Reject Assumption~\ref{assum: biased randomization inference} with the prespecified $\Gamma$ if and only if $\min\{ p_{1\rightarrow 2, \Gamma}, p_{2\rightarrow 1, \Gamma}\}<\alpha/2$.

\end{algorithm}

\subsubsection*{B.3.3: Extending CBT to 1-to-k match}
We cluster $n$ observations into a treated cluster and a control cluster with size $I$ and $I \times K$, respectively, using an appropriate clustering algorithm defined as follows. 

\begin{definition}[Algorithm Appropriateness for Multiple Controls]\label{def: appropriate algorithm multiple contrls}
Let $\mathcal{F} = \{\mathbf{X}, \sum_{j=1}^{K+1} Z_{ij} = 1,\forall i = 1, \cdots, I\}$ be the collection of information. An algorithm $\textsf{ALG}$ is said to be \emph{appropriate} if it takes as input $\mathcal{F}$, and outputs a partition of $I \times K$ matched units into two groups: $\prod_{1}=\{ij_{1}: i=1,\dots,I\}$ and $\prod_{2}=\{ij_{2}: i=1,\dots,I\}$ with $|\prod_1| = I, |\prod_2| = I \times K$, subject to the constraint that there is one and only one treated unit in each matched set $i$.
\end{definition}

The clustering algorithm needs to satisfy the following constraint. 
\begin{constraint}[Multiple-Controls]
In each matched set $i = 1, 2, \cdots, I$, there is one and only one unit from unit $i1,i2,\cdots, i(K+1)$ be clustered to treated cluster, and the other units be clustered to control cluster. 
\label{cons:multi}
\end{constraint}

We now formally state how to test the RA and biased RA using an appropriate algorithm under the CBT framework.

\begin{proposition}[Null distribution with multiple controls]\label{prop: null distribution: multiple controls}
Let $\mathcal{T}=\{ij: Z_{ij}=1, i=1,\dots, I, j=1, \dots, K+1\}$ be the set of indices corresponding to the treated unit in each matched pair. Let $\Pi_1$ be the output from an appropriate algorithm $\textsf{ALG}$ in Definition \ref{def: appropriate algorithm multiple contrls}. Under Assumption~\ref{assum: randomization inference multiple}, we have $|\Pi_{1} \cap  \mathcal{T}|\sim \text{Binomial}\ ( I, 1/(K+1))$.
\end{proposition}

\begin{proof}[Proof of Proposition~\ref{prop: null distribution: multiple controls}]
By Definition~\ref{def: appropriate algorithm multiple contrls} and Assumption~\ref{assum: randomization inference multiple}, we have $\mathbf{1}\{Z_{ij_{1}}=1\}$, $i=1,\dots,I$ are i.i.d. random variables such that $P(\mathbf{1}\{Z_{ij_{1}}=1\}=1\mid \mathbf{X}, \sum_{j=1}^{K+1}Z_{ij}=1)=P(Z_{ij_{1}}=1\mid \mathbf{X}, \sum_{j=1}^{K+1}Z_{ij}=1)=1/(K+1)$. Therefore, conditional on $\{\mathbf{X}, \sum_{j=1}^{K+1}Z_{ij}=1, i=1,\dots,I\}$, we have $|\Pi_{1} \cap  \mathcal{T}|=\sum_{i=1}^{I}\mathbf{1}\{Z_{ij_{1}}=1\}\sim Binomial(I, 1/(K+1))$.
\end{proof}

\begin{proposition}[Bounding one-sided $p$-value with multiple controls]\label{prop: one sided null distribution biased with multiple controls}
Let $\mathcal{T} = \{ij: Z_{ij}=1, i=1,\dots, I, j=1, \dots, K+1\}$ be the set of indices corresponding to the treated units in each matched pair. Let $\Pi_1$ and $\Pi_2$ be the output from an appropriate algorithm $\textsf{ALG}$ in Definition \ref{def: appropriate algorithm multiple contrls}. For $i=1,\dots,I$, define $\overline{T}_{\Gamma i}$ to be independent random variables taking the value $1$ with probability $\Gamma/(K+\Gamma)$ and the value $0$ with probability $K/(K+\Gamma)$, and define $\overline{\overline{T}}_{\Gamma i}$ to be independent random variables taking the value $1$ with probability $1/(1+K\Gamma)$ and the value $0$ with probability $(K\Gamma)/(1+K\Gamma)$. Under Assumption~\ref{assum: biased randomization inference: multiple controls} with $\Gamma \geq 1$, we have for any $t\in \mathbb{R}$,
\begin{align*}
P\left(|\Pi_{1} \cap  \mathcal{T}| \geq t \mid \mathbf{X}, \sum_{j=1}^{K+1}Z_{ij}=1\right)&\leq P\Big(\sum_{i=1}^{I}\overline{T}_{\Gamma i} \geq t \Big)\simeq 1-\Phi\left(\frac{t-\left\{\Gamma/(K + \Gamma)\right\}\cdot I }{ \left\{(K\Gamma)/(K+\Gamma)^2 \cdot I\right\}^{1/2}}\right), \\
 P\left(|\Pi_{1} \cap  \mathcal{T}| \leq t \mid \mathbf{X}, \sum_{j=1}^{K+1}Z_{ij}=1\right)&\leq P\Big(\sum_{i=1}^{I}\overline{\overline{T}}_{\Gamma i} \leq t\Big) \simeq \Phi\left(\frac{t-\left\{1/(1 + K\Gamma)\right\}\cdot I }{ \left\{(K\Gamma)/(1+K\Gamma)^2 \cdot I\right\}^{1/2}}\right),
\end{align*}
where $\Phi(\cdot )$ is the distribution function of standard normal distribution and ``$\simeq$" denotes that two sequences are asymptotically equal as $I \rightarrow \infty$.
\end{proposition}

\begin{proof}[Proof of Proposition~\ref{prop: one sided null distribution biased with multiple controls}]
By Definition~\ref{def: appropriate algorithm multiple contrls} and Assumption~\ref{assum: biased randomization inference: multiple controls}, we have $\mathbf{1}\{Z_{ij_{1}}=1\}$, $i=1,\dots,I$ are i.i.d. random variables such that $P(\mathbf{1}\{Z_{ij_{1}}=1\}=1\mid \mathbf{X}, \sum_{j=1}^{K+1}Z_{ij}=1)=P(Z_{ij_{1}}=1\mid \mathbf{X}, \sum_{j=1}^{K+1}Z_{ij}=1)\in [\frac{1}{1+K\Gamma}, \frac{\Gamma}{K+\Gamma}]$. Therefore, we have for any $t\in \mathbb{R}$,
\begin{align*}
    P\left(|\Pi_{1} \cap  \mathcal{T}| \geq t \mid \mathbf{X}, \sum_{j=1}^{K+1}Z_{ij}=1\right)&=P\left(\sum_{i=1}^{I}\mathbf{1}\{Z_{ij_{1}}=1\} \geq t \mid \mathbf{X}, \sum_{j=1}^{K+1}Z_{ij}=1\right) \\
    &\leq P\Big(\sum_{i=1}^{I}\overline{T}_{\Gamma i} \geq t \Big)\simeq 1-\Phi\left(\frac{t-\left\{\Gamma/(K + \Gamma)\right\}\cdot I }{ \left\{(K\Gamma)/(K+\Gamma)^2 \cdot I\right\}^{1/2}}\right)
\end{align*}
and 
\begin{align*}
    P\left(|\Pi_{1} \cap  \mathcal{T}| \leq t \mid \mathbf{X}, \sum_{j=1}^{K+1}Z_{ij}=1\right)&=P\left(\sum_{i=1}^{I}\mathbf{1}\{Z_{ij_{1}}=1\} \leq t \mid \mathbf{X}, \sum_{j=1}^{K+1}Z_{ij}=1\right)\\
    &\leq P\Big(\sum_{i=1}^{I}\overline{\overline{T}}_{\Gamma i} \leq t\Big)\simeq \Phi\left(\frac{t-\left\{1/(1 + K\Gamma)\right\}\cdot I }{ \left\{(K\Gamma)/(1+K\Gamma)^2 \cdot I\right\}^{1/2}}\right).
\end{align*}
So the desired conclusion follows.
\end{proof}

\begin{proposition}[Bounding two-sided $p$-value with multiple controls]\label{prop: null distribution biased with multiple controls}
Let $\mathcal{T} = \{ij: Z_{ij}=1, i=1,\dots, I, j=1, \dots, K+1\}$ be the set of indices corresponding to the treated units in each matched pair. Let $\Pi_1$ and $\Pi_2$ be the output from an appropriate algorithm $\textsf{ALG}$ in Definition \ref{def: appropriate algorithm multiple contrls}. For $i=1,\dots,I$, define $\overline{T}_{\Gamma i}$ to be independent random variables taking the value $1$ with probability $\Gamma/(K+\Gamma)$ and the value $0$ with probability $K/(K+\Gamma)$, and define $\overline{\overline{T}}_{\Gamma i}$ to be independent random variables taking the value $1$ with probability $1/(1+K\Gamma)$ and the value $0$ with probability $(K\Gamma)/(1+K\Gamma)$. Under Assumption~\ref{assum: biased randomization inference: multiple controls} with $\Gamma \geq 1$, we have for any $t\in \mathbb{R}$,
\begin{align*}
 &\quad P\left((|\Pi_{1} \cap  \mathcal{T}|-I/2)^{2} \geq (t-I/2)^{2} \mid \mathbf{X}, \sum_{j=1}^{K+1}Z_{ij}=1\right)\\
 &\leq  P\Big(\sum_{i=1}^{I}\overline{T}_{\Gamma i} \geq |t-I/2|+I/2 \Big)+P\Big(\sum_{i=1}^{I}\overline{\overline{T}}_{\Gamma i} \leq -|t-I/2|+I/2 \Big) \\
 &\simeq 1-\Phi\left(\frac{|t-I/2|+I/2-\left\{\Gamma/(K + \Gamma)\right\}\cdot I }{ \left\{(K\Gamma)/(K+\Gamma)^2 \cdot I\right\}^{1/2}}\right) + \Phi\left(\frac{-|t-I/2|+I/2-\left\{1/(1 + K\Gamma)\right\}\cdot I }{ \left\{(K\Gamma)/(1+K\Gamma)^2 \cdot I\right\}^{1/2}}\right),
\end{align*}
where $\Phi(\cdot )$ is the distribution function of standard normal distribution and ``$\simeq$" denotes that two sequences are asymptotically equal as $I \rightarrow \infty$.
\end{proposition}

\begin{proof}[Proof of Proposition~\ref{prop: null distribution biased with multiple controls}]
We have for any $t\in \mathbb{R}$,
\begin{align*}
 &\quad P\left((|\Pi_{1} \cap  \mathcal{T}|-I/2)^{2} \geq (t-I/2)^{2} \mid \mathbf{X}, \sum_{j=1}^{K+1}Z_{ij}=1\right)\\
 &=P\left(|\Pi_{1} \cap  \mathcal{T}| \geq |t-I/2|+I/2 \mid \mathbf{X}, \sum_{j=1}^{K+1}Z_{ij}=1\right)+P\left(|\Pi_{1} \cap  \mathcal{T}| \leq -|t-I/2|+I/2 \mid \mathbf{X}, \sum_{j=1}^{K+1}Z_{ij}=1\right)\\
 &\leq  P\Big(\sum_{i=1}^{I}\overline{T}_{\Gamma i} \geq |t-I/2|+I/2 \Big)+P\Big(\sum_{i=1}^{I}\overline{\overline{T}}_{\Gamma i} \leq -|t-I/2|+I/2 \Big) \qquad \qquad \text{(By Proposition~\ref{prop: one sided null distribution biased with multiple controls})} \\
 &\simeq 1-\Phi\left(\frac{|t-I/2|+I/2-\left\{\Gamma/(K + \Gamma)\right\}\cdot I }{ \left\{(K\Gamma)/(K+\Gamma)^2 \cdot I\right\}^{1/2}}\right) + \Phi\left(\frac{-|t-I/2|+I/2-\left\{1/(1 + K\Gamma)\right\}\cdot I }{ \left\{(K\Gamma)/(1+K\Gamma)^2 \cdot I\right\}^{1/2}}\right).
\end{align*}
So the desired conclusion follows. 
\end{proof}

\section*{Supplementary Material C: Details on K-means clustering}
One simple strategy to search for a clustering that respects the side-information leverages a \emph{constrained $K$-means clustering algorithm} with \emph{metric learning}. The conventional, unconstrained $K$-means clustering algorithm partitions $n$ observations into $K$ clusters (\citealp{macqueen1967some, forgy1965cluster}), where $K$ is a pre-specified number of clusters. The algorithm begins by choosing $K$ initial cluster centers, which are then refined iteratively as follows:
\begin{itemize}
    \item[1.] Each observation $\boldsymbol{x}_i$, $i = 1, \dots, n$, is assigned to the cluster with the closest mean (cluster centroid or cluster center).
    \item[2.] Cluster centers $\bm c_j$, $j = 1, \dots, K$, are updated based on the current partition.
\end{itemize}
When the assignment no longer changes, the algorithm converges and clusters $n$ observations into $K$ clusters. In our application, the matched-set-structure side information \ref{side-info: match set structure} further imposes the following \emph{Cannot-Link} constraints to the conventional $K$-means clustering:


\begin{constraint}[Cannot-Link Constraint] \label{Cannot-link contraint}
Unit $i1$ and $i2$ in each matched pair $i = 1, \dots, I$ cannot end up in the same cluster.
\end{constraint}

The conventional $K$-means clustering algorithm plus the \emph{Cannot-Link} constraint leads to the so-called \emph{constrained $K$-means clustering algorithm} originally proposed by \citet{wagstaff2001constrained}. The \emph{Cannot-Link} constraint defines a transitive binary relation over the observations. Therefore, we take a transitive closure over the \emph{Cannot-Link} constraint when we make use of it. In comparison to the conventional $K$-means clustering algorithm, the major modification is that when updating the  cluster memberships, we make sure that the \emph{Cannot-Link} constraints are not violated. 



\begin{algorithm}[ht]
\textbf{Input:} $\{\mathbf{x}_{ij}: i=1,\dots, I, j = 1, 2\}$ and $\mathcal{T} = \{ij, Z_{ij} = 1, i = 1, \dots, I, j = 1, 2 \}$ \;
\textbf{1.} Let $\prod_1 = \{ij_1: i = 1, \dots, I\}$ and $\prod_2 = \{ij_2: i = 1, \dots, I\}$ be initialized clusters, and $\bm c_1$ and $\bm c_2$ centers of $\prod_1$ and $\prod_2$, respectively  \;
\textsf{2.} Assign $ij$ to $\prod_1$ (or $\prod_2$) if $\|\mathbf{x}_{ij} - \bm c_{1}  \|_2 < \|\mathbf{x}_{ij} - \bm c_{2}  \|_2$ (or $\|\mathbf{x}_{ij} - \bm c_{1}  \|_2 > \|\mathbf{x}_{ij} - \bm c_{2}  \|_2$) under the \emph{Cannot-Link} constraint; if the constraint is violated for some $i_*$, assign $i_*1$ to $\prod_1$ and $i_*2$ to $\prod_2$, or vice versa, equiprobably \;
\textsf{3.} For new clusters $\prod_1$ and $\prod_2$, update $\bm c_1$ and $\bm c_2$ \;
\textsf{4.} Iterate 2 and 3 until convergence \;
\textsf{5.} Calculate the test statistic $t = | \prod_1 \bigcap \mathcal{T} | $ \;
\textbf{Output:} The test statistic $t$ and reject Assumption \ref{assum: randomization inference} if $t < c_{\alpha/2}$ or $t > c_{1 - \alpha/2}$ where $c_{\alpha/2}$ and $c_{1-\alpha/2}$ are $\alpha/2$ and $1-\alpha/2$ quantiles of the Binomial$(I,1/2)$ distribution.
\caption{Pseudo Algorithm for Testing Assumption \ref{assum: randomization inference} using CBT }
\label{algo: test assumption 1}
\end{algorithm}

Algorithm \ref{algo: test assumption 1} formally states a two-sided falsification test of the randomization assumption \ref{assum: randomization inference} for matched pair data based on solving a constrained $2$-means clustering algorithm.

 Metric learning is a machine learning technique that automatically constructs task-specific distance metrics from supervised or unsupervised data. The learned distance metric can then be applied to a variety of tasks such as classification and clustering. To be more specific, we consider learning a distance metric of the following form (\citealp{xing2002distance}):

\begin{equation} 
\label{eq: metric}
    d(\bm x, \bm y) = d_{\mathbf{A}}(\bm x, \bm y) = \|\bm x - \bm y \|_{\mathbf{A}} = \sqrt{(\bm x - \bm y)^\top \mathbf{A} (\bm x - \bm y) }
\end{equation}
where $\mathbf{A} \succeq 0$. Setting $\mathbf{A}$ to $\mathbf{I}$ induces the Euclidean distance and $\mathbf{A}$ to the inverse of the variance-covariance matrix of $\mathbf{X}$ induces the Mahalanobis distance. One widely-used metric-learning strategy optimizes the distance metric $d(\boldsymbol x, \boldsymbol y)$ in \eqref{eq: metric} by maximizing the total distances between all ``dissimilar" pairs (i.e., units belonging to different clusters, i.e., units $i1$ and $i2$ in each matched pair $i$ in our case) while enforcing the total distances among ``similar" pairs (i.e., $(\mathbf{x}_l, \bm c_1)$ and $(\mathbf{x}_{l'}, \bm c_2)$ where $\mathbf{x}_l \in \prod_1, \mathbf{x}_{l'} \in \prod_2 $) to be small. The semi-definite programming problem below formalizes this metric-learning approach and learns a matrix $\boldsymbol A$ used to define the distance metric $d(\boldsymbol x, \boldsymbol y)$ in each iteration of the $K$-means clustering algorithm:
\begin{equation} \label{metric learning}
    \begin{split}
        &\max_{\mathbf{A}}~~ \sum_{ i = 1 }^{I} \|\mathbf{x}_{i1} - \mathbf{x}_{i2} \|_{\mathbf{A}}  \\
        \text{subject to}&  \sum_{ \mathbf{x}_l \in  \prod_1 } \|\mathbf{x}_l - \bm c_1  \|_{\mathbf{A}}^2 + \sum_{ \mathbf{x}_{l'} \in  \prod_2 } \|\mathbf{x}_{l'} - \bm c_2  \|_{\mathbf{A}}^2 \leq 1, \\
        &\quad \quad \quad \mathbf{A}  \succeq \mathbf{0},
    \end{split}
\end{equation}
where $\prod_1, \prod_2, \bm c_1, \bm c_2$ are obtained at each iteration of Algorithm \ref{algo: test assumption 1}. Two remarks follow. First, the choice of $1$ in the first constraint is arbitrary; changing it to any positive value $a$ results only in $\mathbf{A}$ being replaced by $a^2\mathbf{A}$. Second, changing $\sum_{ i = 1 }^{I} \|\mathbf{x}_{i1} - \mathbf{x}_{i2} \|_{\mathbf{A}}$ to $\sum_{ i = 1 }^{I} \|\mathbf{x}_{i1} - \mathbf{x}_{i2} \|^2_{\mathbf{A}}$ would result in $\mathbf{A}$ always being rank $1$ and hence not a good option.

To solve the optimization problem \eqref{metric learning}, define
\begin{align*}
    g(\mathbf{A}) = \sum_{ \mathbf{x}_l \in  \prod_1 } \|\mathbf{x}_l - \bm c_1  \|_{\mathbf{A}}^2 + \sum_{ \mathbf{x}_{l'} \in  \prod_2 } \|\mathbf{x}_{l'} - \bm c_2  \|_{\mathbf{A}}^2 - \log \left( \sum_{ i = 1 }^{I} \|\mathbf{x}_{i1} - \mathbf{x}_{i2} \|_{\mathbf{A}}   \right).
\end{align*}
Solving (\ref{metric learning}) is equivalent to minimize $g(\mathbf{A})$. For the case of diagonal matrix $\mathbf{A}$, we can apply the Newton-Raphson method; otherwise, the gradient descent and iterative projections would be better choices as Newton's method is prohibitively expensive (\citealp{boyd2004convex}).

A metric-learning-enhanced version of Algorithm \ref{algo: test assumption 1} replaces Step $2$ with the following:
\begin{enumerate}
    \item[$2^\ast.$] Learn the positive semi-definite matrix $\boldsymbol A$ by solving the optimization problem \eqref{metric learning}. Denote the associated distance metric by $\|\cdot\|_{\mathbf{A}}$. Assign $ij$ to $\prod_1$ (or $\prod_2$) if $\|\mathbf{x}_{ij} - \bm c_{1}  \|_{\mathbf{A}} < \|\mathbf{x}_{ij} - \bm c_{2}  \|_{\mathbf{A}}$ (or $\|\mathbf{x}_{ij} - \bm c_{1}  \|_{\mathbf{A}} > \|\mathbf{x}_{ij} - \bm c_{2}  \|_{\mathbf{A}}$) under the \emph{Cannot-Link} constraint; if the constraint is violated for some $i_*$, assign $i_*1$ to $\prod_1$ and $i_*2$ to $\prod_2$, or vice versa, equiprobably.
\end{enumerate}

We may impose further restrictions on matrix $\mathbf{A}$ to increase the interpretability of the distance metric $\|\cdot\|_\mathbf{A}$. Motivated by \citet{diamond2013genetic}, one may consider $\mathbf{A}$ of the following form:
\begin{align*}
    \mathbf{A} = (\mathbf{S}^{-1/2})^\top \mathbf{D} \mathbf{S}^{-1/2},  
\end{align*}
where $\mathbf{D}$ is a diagonal matrix to be learned and $\mathbf{S}^{1/2}$ is the Cholesky decomposition of the variance-covariance matrix $\mathbf{S}$. According to this specification, a large diagonal entry of $\mathbf{D}$ corresponds to the distance metric giving a heavy weight to the corresponding normalized covariate when performing the clustering, and should be taken as a signal of poor covariate balance of the corresponding normalized covariate.

Although we have focused on the metric learning algorithm from \citet{xing2002distance}, other metric learning algorithms such as information theoretic metric learning (\citealp{davis2007information}), sparse high-dimensional metric learning (\citealp{qi2009efficient}), and relative components analysis (\citealp{shental2002adjustment, bar2003learning, bar2005learning}) could also be fused with the $K$-means clustering algorithm to enhance the performance.

\section*{Supplementary Material D: Additional Simulation Results}
We describe additional simulation results concerning the performance of outcome analysis that takes into account RSVs. We considered the data generating process as detailed in Section \ref{sec: simulation} with $n = 3000$, $c = 0.3$, and the following nonlinear outcome models for $R_C$:
\[
R_C = \kappa_2 \times \textsf{sign}\{X_1\} \cdot |X_1|^{\kappa_1} + 0.5\sqrt{|X_2|} - X_3 + \epsilon,\qquad \epsilon \sim \text{Normal}(0, 1),
\] with different choices of $\kappa_1$ and $\kappa_2$, and let $\beta = 0$ so that $R_T = R_C$ for all study units and there is no treatment effect. We used the \textsf{senWilcox} function in the \textsf{R} package \textsf{DOS} to conduct outcome analysis for various matched datasets. Under a biased randomization scheme, i.e., Assumption 2 in the main article, the treatment effect is only partially identified, i.e., we obtain a partial identification interval rather than a point estimate, and this interval, unlike a confidence interval, does not shrink to $0$ even as sample size goes to infinity. Therefore, to still measure the performance of the partial identification interval, we define the ``bias" of the interval as the minimal distance between the estimand to the interval. For instance, if the partial identification interval is $[-0.02, 0.03]$, then the ``bias" is equal to $0$ because it contains the estimand $\beta = 0$. As another example, if the partial identification interval is $[0.02, 0.03]$, then the bias is defined to be $0.02$ as this is the minimal distance from $0$ to the interval $[0.02, 0.03]$. We also report the coverage probability of the confidence intervals that are well-defined for both under point ($\Gamma = 1$) and partial ($\Gamma > 1$) identification. 

Table \ref{tbl: RI and biased RI level} summarizes the bias and coverage probability of an outcome analysis that assumes randomization (i.e., $\Gamma = 1$) on matched-pair data, and an outcome analysis that takes into account the RSV $\underline{\widetilde{\Gamma}}$ obtained using $\textsf{SS-CPT}_{\textsf{pscore}}$ and conducts the outcome analysis under $\Gamma = \underline{\widetilde{\Gamma}}$, as advocated in the main article. We found several consistent trends. First, when $\kappa_1$ and $\kappa_2$ are not both $0$, there is selection bias and statistical matching helps adjust for covariate and remove overt bias. Optimal matching ($\mathcal{M}_{\textsf{opt}}$) seemed to consistently outperform the metric-based matching ($\mathcal{M}_{\textsf{maha}}$) and propensity score matching ($\mathcal{M}_{\textsf{pscore}}$), and removed most of the bias. The testing procedure $\textsf{SS-CPT}_{\textsf{pscore}}$ has little power against optimally-matched datasets (i.e., the RSV $\underline{\widetilde{\Gamma}}$ is only occasionally greater than $1$; see Table \ref{tbl: simulation results main article} in the main article), and the coverage probabilities under both  $\Gamma = 1$ and $\Gamma = \underline{\widetilde{\Gamma}}$ are similar and close to the nominal rate. Second, $\mathcal{M}_{\textsf{maha}}$ and $\mathcal{M}_{\textsf{pscore}}$ tended to remove less bias and therefore have poor coverage probability under $\Gamma = 1$; however, the partial identification interval obtained under $\Gamma = \underline{\widetilde{\Gamma}}$ in fact contains the true estimand $\beta = 0$ in almost all of simulated datasets, and the bias (defined as the distance from the estimand to the partial identification interval) becomes negligible. Moreover, the confidence intervals corresponding to the partial identification intervals always obtain nominal level, although very conservative, due to the conservative nature of partial identification and Rosenbaum bounds (i.e., Rosenbaum bounds is a worst-case $p$-value).

\begin{table}[ht]
\centering
\caption{Simulation results of performing the outcome analysis on matched data assuming randomization ($\Gamma = 1$) versus further taking into account the RSV ($\Gamma = \underline{\widetilde{\Gamma}}$).}
\label{tbl: RI and biased RI level}
\begin{tabular}{ccccccc}
  \hline
 \multirow{3}{*}{\begin{tabular}{c}$\kappa_2$ \end{tabular}} & \multirow{3}{*}{\begin{tabular}{c}$\kappa_1$ \end{tabular}} & \multirow{3}{*}{\begin{tabular}{c}$\mathcal{M}$ \end{tabular}} & \multirow{3}{*}{\begin{tabular}{c}$100\times$ Bias\\ $\Gamma = 1$\end{tabular}} &
 \multirow{3}{*}{\begin{tabular}{c}$100 \times$ Bias\\ $\Gamma = \underline{\widetilde{\Gamma}}$\end{tabular}} &
 \multirow{3}{*}{\begin{tabular}{c}Coverage \\ $\Gamma = 1$ \end{tabular}}  & \multirow{3}{*}{\begin{tabular}{c}Coverage \\ $\Gamma = \underline{\widetilde{\Gamma}}$ \end{tabular}}   \\ \\ \\
  \hline
   0.10 & 0.00 & $\mathcal{M}_{\textsf{opt}}$ & 2.06 & 2.03 & 94.6\% & 94.6\% \\ 
  0.10 & 0.00 & $\mathcal{M}_{\textsf{maha}}$ & 2.70 & 0.15 & 89.6\% & 99.8\% \\ 
  0.10 & 0.00 & $\mathcal{M}_{\textsf{pscore}}$ & 2.99 & 0.04 & 93.4\% & 100\% \\ 
    0.10 & 0.10 & $\mathcal{M}_{\textsf{opt}}$ & 1.90 & 1.89 & 93.8\% & 93.8\% \\ 
  0.10 & 0.10 & $\mathcal{M}_{\textsf{maha}}$ & 2.69 & 0.13 & 90.6\% & 100\% \\ 
  0.10 & 0.10 & $\mathcal{M}_{\textsf{pscore}}$ & 2.23 & 0.03 & 94.4\% & 100\% \\ 
   0.10 & 0.20 & $\mathcal{M}_{\textsf{opt}}$ & 2.22 & 2.20 & 92.8\% & 92.8\% \\ 
  0.10 & 0.20 & $\mathcal{M}_{\textsf{maha}}$ & 2.70 & 0.11 & 93.4\% & 100\% \\ 
  0.10 & 0.20 & $\mathcal{M}_{\textsf{pscore}}$ & 2.76 & 0.02 & 93.8\% & 100\% \\ 
   0.20 & 0.00 & $\mathcal{M}_{\textsf{opt}}$ & 1.96 & 1.96 & 95.2\% & 95.2\% \\ 
  0.20 & 0.00 & $\mathcal{M}_{\textsf{maha}}$ & 4.09 & 0.22 & 85.0\% & 99.4\% \\ 
  0.20 & 0.00 & $\mathcal{M}_{\textsf{pscore}}$ & 5.07 & 0.17 & 89.4\% & 100\% \\ 
    0.20 & 0.10 & $\mathcal{M}_{\textsf{opt}}$ & 2.05 & 2.06 & 94.0\% & 94.0\% \\ 
  0.20 & 0.10 & $\mathcal{M}_{\textsf{maha}}$ & 3.73 & 0.18 & 88.6\% & 99.8\% \\ 
  0.20 & 0.10 & $\mathcal{M}_{\textsf{pscore}}$ & 5.10 & 0.08 & 88.8\% & 100\% \\ 
   0.20 & 0.20 & $\mathcal{M}_{\textsf{opt}}$ & 2.42 & 2.40 & 92.6\% & 92.6\% \\
  0.20 & 0.20 & $\mathcal{M}_{\textsf{maha}}$ & 3.66 & 0.23 & 88.4\% & 99.4\% \\ 
  0.20 & 0.20 & $\mathcal{M}_{\textsf{pscore}}$ & 5.20 & 0.10 & 87.4\% & 100\% \\
   \hline
\end{tabular}
\end{table}


\end{document}